%% file: lmcs16.tex
\documentclass{CSML}
\pdfoutput=1

\usepackage{lastpage}
\newcommand{\dOi}{14(2:16)2018}
\lmcsheading{}{1--\pageref{LastPage}}{}{}%
{Jul.~18, 2017}{Jun.~22, 2018}{}

\usepackage{hyperref}
\hypersetup{hidelinks}

\usepackage{xspace}

\usepackage{tikz}
\usetikzlibrary{snakes,arrows,shapes,er,positioning}
\usepackage{amsmath,amssymb,amsfonts,mathrsfs,latexsym,stmaryrd}
\usepackage{pseudocode}
\usepackage{color}
\usepackage{multirow}
\usepackage{graphicx}
\usepackage{ifpdf}
\usepackage[english]{babel}
\usepackage[noend]{algpseudocode}
\usepackage{algorithm}
\usepackage{balance}

\newenvironment{compactenum}{\begin{enumerate}[nosep]}{\end{enumerate}}
\newenvironment{compactitem}{\begin{itemize}[nosep]}{\end{itemize}}
\newenvironment{inparaenum}{\begin{enumerate*}[label={\emph{(\roman*)}}]}{\end{enumerate*}}

\newtheoremstyle{cited}%
  {}
  {}
  {\itshape}
  {}
  {\bfseries}
  {\bf .}
  {.5em}
  {\thmname{#1} \thmnumber{#2} \thmnote{\normalfont#3}}
\theoremstyle{cited}
\newtheorem{citedthm}[thm]{Theorem}
\newtheorem{citedlem}[thm]{Lemma}

\input{macro}

\begin{document}

\title[On Sub-Propositional Fragments of Modal Logic]{On Sub-Propositional Fragments of Modal Logic}

\author[D. Bresolin et al.]{Davide Bresolin}	
\address{University of Padova, Italy}	
\email{davide.bresolin@unipd.it}  

\author[]{Emilio Mu\~noz-Velasco}	
\address{Universidad de M\'alaga, Spain}	
\email{ejmunoz@uma.es}  

\author[]{Guido Sciavicco}	
\address{University of Ferrara, Italy }	
\email{guido.sciavicco@unife.it}  



\subjclass{Theory; Verification; Algorithms.}


\begin{abstract}
In this paper, we consider the well-known modal logics \Kmono, \Tmono, \Fmono, and \SFmono, and we study some of their sub-propositional fragments, namely the classical Horn fragment, the Krom fragment, the so-called core fragment, defined as the intersection of the Horn and the Krom fragments, plus their sub-fragments obtained by limiting the use of boxes and diamonds in clauses. We focus, first, on the relative expressive power of such languages: we introduce a suitable measure of expressive power, and we obtain a complex hierarchy that encompasses all fragments of the considered logics. Then, after observing the low expressive power, in particular, of the Horn fragments without diamonds, we study the computational complexity of their satisfiability problem, proving that, in general, it becomes polynomial.
\end{abstract}

\maketitle

\section{Introduction}\label{section:introduction}

Modal logic has been applied to computer science in several areas, such as artificial intelligence, distributed systems, and computational linguistics, and modal operators have been interpreted in different ways for different applications. %
Moreover, modal logic is paradigmatic of the whole variety of description logics~\cite{Baader:2003:DLH:885746}, temporal logics~\cite{temporal_logic_foundations}, spatial logics~\cite{DBLP:reference/spatial/2007}, among others. It is well-known that the basic (normal) modal logic  \Kmono, can be {\em semantically} defined as the logic of all directed relations, or {\em syntactically} defined via the single axiom $K$: $\Box(p\rightarrow q)\rightarrow (\Box p\rightarrow \Box q)$. Modal logics of particular classes of relations are often referred to as {\em axiomatic extensions} of \Kmono. In this paper we consider, besides \Kmono, some relevant extensions such as the logic \Fmono\ of all {\em transitive} relations (defined by the axiom $4$: $\Box p\rightarrow \Box\Box p$), the logic \Tmono\ of all {\em reflexive} relations (defined by the axiom $T$: $\Box p\rightarrow p$), and the logic \SFmono of all reflexive and transitive relations (preorders), defined by the combination of all axioms above. From a purely semantic perspective, these logics (together with \SFive, that is, the logic of all transitive, reflexive and symmetric relations) can be represented in a lattice as in Fig.~\ref{fig:logics}. The satisfiability problem for these modal logics is \PSpace-complete for all these logics, except for \SFive, for which it becomes \NP-complete~\cite{DBLP:journals/siamcomp/Ladner77}.

\begin{figure}[t]
\begin{center}
\begin{tikzpicture}[-,node distance=1.5cm,on grid]
\node[draw=none]   (T) {\Tmono};
\node[draw=none]   (K)[above right = of T] {\Kmono};
\node[draw=none]   (S4)[below right = of T] {\SFmono};
\node[draw=none]   (K4)[below right = of K] {\Fmono};
\node[draw=none]   (S5)[below = of S4] {\SFive};
\draw[-] (K) -- (T);
\draw[-] (K) -- (K4);
\draw[-] (K4) -- (S4);
\draw[-] (T) -- (S4);
\draw[-] (S4) -- (S5);
\end{tikzpicture}\label{fig:logics}
\end{center}
\caption{Classical axiomatic extensions of \Kmono\ and their semantical containment relation.}
\end{figure}

\medskip

The quest of computationally well-behaved restrictions of logics, and the observation that meaningful statements can be still expressed under sub-propositional restrictions has recently increased the interest in sub-propositional fragments of temporal description logics~\cite{DBLP:journals/tocl/ArtaleKRZ14}, temporal logics~\cite{DBLP:conf/lpar/ArtaleKRZ13}, and interval temporal logics~\cite{artale2015,DBLP:conf/jelia/BresolinMS14}. However, sub-propositional modal logics have received little or no attention, with the exception of~\cite{ChenLin94,del1987note,nguyen2004complexity}, which are limited to the Horn fragment. There are two standard ways to weaken the classical propositional language based on the clausal form of formulas: the {\em Horn fragment}, that only allows clauses with at most one positive literal~\cite{horn}, and the {\em Krom fragment}, that only allows clauses with at most two (positive or negative) literals~\cite{krom}. The {\em core fragment} combines both restrictions. Orthogonally, one can restrict a modal language in clausal form by allowing only diamonds or only boxes in positive literals, obtaining weaker fragments that we call, respectively, the {\em diamond fragment} and {\em box fragment}. By combining these two levels of restriction, one may obtain several sub-propositional fragments of modal logics, and, by extension, of description, temporal, and spatial logics.

\medskip

The satisfiability problem for classical propositional Horn logic is $\Pt$-complete~\cite{cookhorn}, while the classical propositional Krom logic satisfiability problem (also known as the 2-SAT problem) is $\NLOG$-complete~\cite{papa2003}, and the same holds for the core fragment (this results is an immediate strengthening of the classical $\NLOG$-hardness of 2-SAT). The satisfiability problem for quantified propositional logic (\QBF), which is \PSpace-complete in its general form, becomes \Pt\ when formulas are restricted to binary (Krom) clauses~\cite{asp79}. In~\cite{DBLP:conf/lpar/ArtaleKRZ13,ChenLin93}, the authors study different sub-propositional fragments of Linear Temporal Logic (\LTL). By excluding the operators Since and Until from the language, and keeping only the Next/Previous-time operators and the box version of Future and Past, it is possible to prove that the Krom and core fragments are \NP-hard, while the Horn fragment is still \PSpace-complete (the same complexity of the full language). Moreover, the complexity of the Horn, Krom, and core fragments without Next/Previous-time operators range from \NLOG\ (core), to \Pt\ (Horn), to \NP-hard (Krom). Where only a universal (anywhere in time) modality is allowed their complexity is even lower (from \NLOG\ to \Pt). Temporal extensions of the description logic DL-Lite have been studied under similar sub-propositional restrictions, and similar improvements in the complexity of various problems have been found~\cite{DBLP:journals/tocl/ArtaleKRZ14}. Sub-propositional fragments of the undecidable interval temporal logic \HS~\cite{interval_modal_logic}, have also been studied. The Horn, Krom, and core restrictions of \HS\ are still undecidable~\cite{DBLP:conf/jelia/BresolinMS14}, but weaker restrictions have shown positive results. In particular, the Horn fragment of \HS\ without diamonds becomes $\Pt$-complete in two interesting cases~\cite{artale2015,bresolin2017horn}: first, when it is interpreted over dense linear orders, and, second, when the semantics of its modalities are made reflexive. On the bases of these results, sub-propositional interval temporal extensions of description logics have been introduced in~\cite{artale2015}. Among other clausal forms of (temporal) logics, it is worth mentioning the fragment \Gr\ of \LTL~\cite{DBLP:journals/jcss/BloemJPPS12}. Lastly, the Horn fragments of modal logic \Kmono and of several of its axiomatic extensions have been considered in~\cite{ChenLin94,del1987note,nguyen2004complexity}; in particular, it is known that \Kmono, \Tmono, \Fmono, and \SFmono are all \PSpace-complete even under the Horn restriction, while \SFive, which is \NP-complete in the full propositional case, becomes \Pt-complete.

\medskip

In this paper, we consider sub-propositional fragments of the modal logics \Kmono, \Tmono, \Fmono, and \SFmono, and study their relative expressive power in a systematic way. We consider two different notions of relative expressive power for fragments of modal logic, and we provide several results that give rise to two different hierarchies, leaving only a few open problems. Then, motivated by the observation that the fragment Horn box (and the fragment core box) of the logics we consider are expressively very weak, we launch an investigation about the complexity of their satisfiability problem, and we prove a suitable small-model theorem for such fragments. When paired with the results in~\cite{Nguyen2000,Nguyen2016}, such a small-model result allows us to prove \Pt-completeness of the Horn box fragments of \Kmono, \Tmono, \Fmono, and \SFmono.
Finally, we give a deterministic polynomial time satisfiability-checking algorithm that covers all considered fragments and that it is easy to implement.
To the best of our knowledge, this is the first work where sub-Krom and sub-Horn fragments of \Kmono have been considered.

\medskip

Weak logics such as those mentioned in this paper can be applied in practice, but the main motivation to study them is that they are basic modal logics, and many other (practically useful) formalisms are built directly over them. Among these, we may mention {\em epistemic} logic, {\em deontic} logic, and {\em provability} logic~\cite{Gabbay:2013:HPL:2564760} (the latter, in particular, is a known variant of \SFmono), but, also, the entire range of description logics (starting from $\mathcal{\bf ALC}$, which is syntactic variant of \Kmono~\cite{vanHarmelen:2007:HKR:1557461}). Studying sub-propositional fragments, such as Horn-like fragments of logics is naturally motivated by the fact that Horn-like formulas can be seen as rules, and this is the basis of logic programming. {\em Modal} logic programming is a relatively important topic (see, e.g.,~\cite{Barringer1990,debart92,DBLP:journals/fuin/Nguyen03}). While SLD-resolution can be adapted to several modal logics under the Horn restriction~\cite{DBLP:journals/fuin/Nguyen03}, it is natural to imagine that a similar rule can be devised under the Horn box restriction that allow even more efficient deductions --- we briefly consider this problem in the last section of this paper. Modal logic programs can be used in modal deductive databases (which consist of facts and positive modal logic rules)~\cite{modal_databases}, and one may be interested in the fact that certain types of queries belong to a fragment of modal logic with better computational properties.

\medskip

The paper is organized as follows. After the necessary preliminaries, in Section 2, we study the relative expressive power of the various fragments in Section 3, and the complexity of the satisfiability problem for the Horn box fragments is studied in Section 4, before concluding.

\begin{figure*}[t!]
   \centering
\begin{tikzpicture}[>=latex,line join=bevel,scale=0.60,thick]
\begin{scope}
  \pgfsetstrokecolor{black}
  \definecolor{strokecol}{rgb}{1.0,0.0,0.0};
  \pgfsetstrokecolor{strokecol}
\end{scope}
\begin{scope}
  \pgfsetstrokecolor{black}
  \definecolor{strokecol}{rgb}{1.0,0.0,0.0};
  \pgfsetstrokecolor{strokecol}
\end{scope}
  \node (Core) at (336.0bp,106.0bp) [draw,draw=none] {$\LCore$};
  \node (Bool) at (335.0bp,250.0bp) [draw,draw=none] {$\LBool$};
  \node (CoreBox) at (285.0bp,34.0bp) [draw,draw=none] {$\LCoreBox$};
  \node (Krom) at (177.0bp,178.0bp) [draw,draw=none] {$\LKrom$};
  \node (HornBox) at (500.0bp,106.0bp) [draw,draw=none] {$\LHornBox$};
  \node (Horn) at (494.0bp,178.0bp) [draw,draw=none] {$\LHorn$};
  \node (KromDia) at (172.0bp,106.0bp) [draw,draw=none] {$\LKromDia$};
  \node (HornDia) at (606.0bp,106.0bp) [draw,draw=none] {$\LHornDia$};
  \node (KromBox) at (63.0bp,106.0bp) [draw,draw=none] {$\LKromBox$};
  \node (CoreDia) at (389.0bp,34.0bp) [draw,draw=none] {$\LCoreDia$};
  \draw [->,solid] (Krom) ..controls (175.21bp,151.98bp) and (174.55bp,142.71bp)  .. (KromDia);
  \draw [->,solid] (Horn) ..controls (496.14bp,151.98bp) and (496.94bp,142.71bp)  .. (HornBox);
  \draw [->,solid] (Horn) ..controls (535.91bp,150.81bp) and (553.92bp,139.55bp)  .. (HornDia);
  \draw [->,solid] (KromDia) ..controls (215.2bp,85.116bp) and (222.26bp,82.327bp)  .. (229.0bp,80.0bp) .. controls (273.42bp,64.667bp) and (288.24bp,65.91bp)  .. (CoreDia);
  \draw [->,solid] (HornDia) ..controls (566.65bp,85.093bp) and (560.19bp,82.312bp)  .. (554.0bp,80.0bp) .. controls (517.07bp,66.216bp) and (474.06bp,54.626bp)  .. (CoreDia);
  \draw [->,solid] (Bool) ..controls (279.09bp,224.23bp) and (246.9bp,209.97bp)  .. (Krom);
  \draw [->,solid] (KromBox) ..controls (104.67bp,85.136bp) and (111.49bp,82.341bp)  .. (118.0bp,80.0bp) .. controls (155.13bp,66.659bp) and (198.15bp,55.157bp)  .. (CoreBox);
  \draw [->,solid] (Krom) ..controls (234.17bp,151.83bp) and (263.35bp,138.99bp)  .. (Core);
  \draw [->,solid] (Horn) ..controls (438.57bp,152.44bp) and (408.8bp,139.25bp)  .. (Core);
  \draw [->,solid] (Core) ..controls (317.4bp,79.474bp) and (310.12bp,69.483bp)  .. (CoreBox);
  \draw [->,solid] (Bool) ..controls (391.5bp,224.12bp) and (425.12bp,209.32bp)  .. (Horn);
  \draw [->,solid] (Krom) ..controls (134.21bp,150.72bp) and (115.67bp,139.34bp)  .. (KromBox);
  \draw [->,solid] (Core) ..controls (355.39bp,79.389bp) and (363.05bp,69.277bp)  .. (CoreDia);
  \draw [->,solid] (HornBox) ..controls (456.57bp,85.077bp) and (449.63bp,82.322bp)  .. (443.0bp,80.0bp) .. controls (400.18bp,64.998bp) and (386.05bp,65.476bp)  .. (CoreBox);
%
%
\draw [dashed] plot [smooth] coordinates {(280bp,240bp) (380bp,140bp) (700bp,140bp) };
\draw [dashed] plot [smooth] coordinates {(230bp,20bp) (370bp,190bp) (690bp,220bp) };
\draw [dashed] plot [smooth] coordinates {(250bp,260bp) (320bp,220bp) (690bp,235bp) };
\node[rotate=2.5] at (630bp,155bp) {\begin{tabular}{l}\footnotesize{$\PSpace$-complete}~\cite{ChenLin94}\\ \footnotesize{$\mathbf L=\Kmono,\Tmono,\Fmono,\SFmono$}\end{tabular}};
\node[rotate=2.5] at (590bp,207bp) {$\in\Pt$~\cite{nguyen2004complexity}, $\mathbf L=\SFive$};
\node[rotate=2.5] at (570bp,245bp) {$\NP$-complete~\cite{nguyen2004complexity}, $\mathbf L=\SFive$};
\end{tikzpicture}
	\caption{Relative expressive power by containment and known complexity results.}
   \label{fig:exprpower}
\end{figure*}

\section{Syntax and Semantics}\label{section:preliminaries}

Let us fix a unary modal {\em similarity type} with a single modality and a denumerable set
$\mathcal P$ of propositional letters. The associated {\em modal language} \Kmono contains all and only the formulas generated by the following grammar:
\begin{equation}
\varphi::=\top\mid p\mid \neg\varphi\mid \varphi\vee\varphi\mid \diax\varphi\mid \boxx\varphi,\label{boolgrammar}
\end{equation}

\noindent where $p\in\mathcal P$, $\diax$ is called {\em diamond}, and $\boxx$ is called {\em box}. The other classical operators, such as $\rightarrow$ and $\wedge$, can be considered as abbreviations. The \emph{length} of a formula $\varphi$ is denoted by $|\varphi|$ and it is defined as the number of symbols in $\varphi$. By $md(\varphi)$, we denote the {\em modal depth} of $\varphi$, that is, the maximum number of nested boxes and diamonds in $\varphi$. Finally, $Cl(\varphi)$ denotes the \emph{closure} of the formula $\varphi$, that is, the set of all sub-formulas of $\varphi$. Notice that both the modal depth and the cardinality of the closure are bounded by $|\varphi|$.

\medskip

A {\em Kripke} {\em frame} is a relational structure $\mathcal F=(W,R)$, where the elements of $W\neq\emptyset$ are called {\em possible worlds}, and $R\in W\times W$ is the {\em accessibility} relation.
In the following, for any given accessibility relation $R$, we denote by $R^\circlearrowright$ its \emph{reflexive} closure, by $\overrightarrow R$ its \emph{transitive closure}, and by $R^*$ its \emph{reflexive and transitive closure}.
A {\em Kripke structure} over the frame $\mathcal F$ is a pair $M=(\mathcal F,V)$, where $V:W\rightarrow 2^\mathcal P$ is an {\em valuation function}, and we say that $M$ {\em models} $\varphi$ at the world $w$, denoted by $M,w\mmodels\varphi$, if and only if either one of the following holds:

\medskip

\begin{compactitem}
\item $\varphi=\top$;
\item $\varphi = p$, and $p\in V(w)$;
\item $\varphi= \neg\psi$, and $M,w\not\mmodels\varphi$;
\item $\varphi= \psi\vee\xi$, and $M,w\mmodels\psi$ or $M,w\mmodels\xi$;
\item $\varphi= \diax\psi$, and there exists $v$ such that $wR v$ and $M,v\mmodels\psi$;
\item $\varphi= \boxx\psi$, and for every $v$ such that $wR v$, it is the case that $M,v\mmodels\psi$.
\end{compactitem}

\medskip

%
%

\noindent In this case, we say that $M$ is a {\em model} of $\varphi$, and that $\varphi$ is {\em satisfiable}; in the following, we (improperly) use the terms model and structure as synonyms. We define the \emph{size} $|M|$ of the model $M$ as the cardinality of the set of worlds $W$.

\medskip

The modal language can be interpreted in classes of Kripke frames with specific properties. Notable examples are: the class of all {\em reflexive} Kripke frames, the class of all {\em transitive} Kripke frames, the class of all {\em reflexive and transitive} Kripke frames, and the class of all Kripke frames whose relation is an {\em equivalence} (that is, it is reflexive, transitive, and symmetric). The corresponding logics are usually called \Tmono, \Fmono, \SFmono, and \SFive, respectively, and the notion of satisfiability in \Tmono, \Fmono, \SFmono, and \SFive is modified accordingly.
%

\medskip

In order to define sub-propositional fragments of a modal logic $\mathbf L$ we start from the {\em clausal form} of $\mathbf L$-formulas, whose building blocks are the {\em positive literals}:
\begin{eqnarray}
\lambda  &::=& \top \mid p\mid \diax \lambda \mid \boxx \lambda,\label{poslit}
\end{eqnarray}

\noindent and we say that $\varphi$ is in {\em clausal form} if it can be generated by the following grammar:
\begin{eqnarray}
\varphi  &::=& \lambda \mid \neg\lambda\mid \varphi \land \varphi\mid\nabla (\neg\lambda_1\vee\ldots\vee\neg\lambda_n\vee\lambda_{n+1}\vee\ldots\vee\lambda_{n+m})\label{clause}
\end{eqnarray}

\noindent where $\nabla=\underbrace{\Box\Box\ldots}_{s}$\,, $s\ge 0$ and $\nabla (\neg\lambda_1\vee\ldots\vee\neg\lambda_n\vee\lambda_{n+1}\vee\ldots\vee\lambda_{n+m})$ is a \emph{clause}.

In the following, we use $\varphi_1,\varphi_2,\ldots$ to denote clauses, $\lambda_1,\lambda_2,\ldots$ to denote positive literals, and $\varphi,\psi,\xi,\ldots$ to denote formulas. We write clauses in their implicative form $\nabla(\lambda_1\land\ldots\land\lambda_n\rightarrow\lambda_{n+1}\vee\ldots\vee\lambda_{n+m})$, and use $\bot$ as a shortcut for $\neg\top$.
Moreover, since positive and negative literals $\lambda$ and $\neg\lambda$ are equivalent to the clauses $\top\rightarrow\lambda$ and $\lambda\rightarrow\bot$, respectively, from now on we assume that  any formula produced by (\ref{clause}) is a conjunction of clauses $\varphi_1\wedge\varphi_2\wedge\ldots\wedge\varphi_l$, unless specified otherwise.
It is known that every modal logic formula can be rewritten in clausal form~\cite{nguyen2004complexity}.

\begin{citedthm}[\cite{nguyen2004complexity}]\label{boolvsclause}
For every modal logic formula $\varphi$ there exists an equisatisfiable conjunction of clauses $\varphi_1\wedge\varphi_2\wedge\ldots\wedge\varphi_l$.
\end{citedthm}


{\em Sub-propositional} fragments of a modal logic $\mathbf L$ can be now defined by constraining the cardinality and the structure of clauses: the fragment of $\mathbf L$ in clausal form where each clause in (\ref{clause}) is such that $m\le 1$ is called {\em Horn} fragment, and denoted by \LHorn, and when each clause is such that $n+m\le 2$ it is called {\em Krom} fragment, and it is denoted by \LKrom; when both restrictions apply we denote the resulting fragment, the {\em core} fragment, by \LCore. It is also interesting to study the fragments that can be obtained from by allowing only the use of \boxx\ or \diax\ in positive literals, that are generically called the {\em box} and {\em diamond} fragments of modal logic $\mathbf L$.
In this way, we define \LHornBox and \LHornDia as, respectively, the box and the diamond fragments of \LHorn. By applying the same restrictions to \LKrom and \LCore, we obtain the pair \LKromBox and \LKromDia from the former, and the pair \LCoreBox and \LCoreDia, from the latter. Obviously, it makes no sense to consider the fragments $\mathbf{L}^\boxx$ and $\mathbf{L}^\diax$, as boxes are diamonds are mutually definable when the full propositional power is allowed.

\medskip

It should be noted that in the literature there is no unified definition of the different modal or temporal sub-propositional logics. Our definition follows Nguyen~\cite{nguyen2004complexity}, with a notable difference: while the definition of clauses is the same, we choose a more restrictive definition of what is a formula. Hence, a formula of \LHorn by our definition is also a Horn formula by~\cite{nguyen2004complexity}, but not vice versa. However, since every Horn formula by the definition of Nguyen can be transformed into a conjunction of Horn clauses, the two definitions are equivalent. The definitions of~\cite{ChenLin94,del1987note} are equivalent to that of Nguyen, and hence to our own.
Other approaches force clauses to be quantified using a {\em universal} modality that asserts the truth of a formula in every world of the model. The universal modality is either assumed in the language~\cite{DBLP:conf/lpar/ArtaleKRZ13} or it is definable using the other modalities~\cite{bresolin2017horn,DBLP:conf/jelia/BresolinMS14}, but the common choice in the literature of modal (non-temporal) logic is simply excluding the universal modality. However, most of our results in expressive power (see the next section) still hold when clauses are universally quantified; the question of whether our complexity results hold as well, on the contrary, is still open.

\medskip

Among all sub-propositional fragments that emerge from the above discussion, we shall see that one is particularly interesting, that is, \LHornBox, which presents the sharpest improvement in complexity of the satisfiability problem when compared to the original language $\mathbf L$. It is worth to notice that such fragment is extremely natural: in fact, it naturally corresponds to first order Horn logic. To see this, consider the clause

$$p\rightarrow\diax q,$$

\noindent which is Horn, but not Horn box. It translates to the first order clause

$$(\forall x)(P(x)\rightarrow\exists y(xRy\wedge Q(y))),$$

\noindent which is not Horn, nor it has a first order Horn equivalent. On the other hand, the clause

$$p\rightarrow\boxx q,$$

\noindent which is Horn box, translates to the first order clause

$$(\forall x)(P(x)\rightarrow\forall y(xRy\rightarrow Q(y))),$$

\noindent that is,

$$(\forall x)\forall y((P(x)\wedge xRy)\rightarrow Q(y)),$$

\noindent which is a first order Horn clause.

\section{Relative Expressive Power}\label{section:expressivity}


In this section, we study the relative expressive power of $\Lmono\in\{\Kmono,\Tmono, \Fmono, \SFmono\}$ and their fragments. The simplest way to compare the relative expressive power of such fragments is by syntactical containment (see Fig.~\ref{fig:exprpower}, where $\LBool$ stands for $\mathbf L$ to stress that we mean the unconstrained language), but, obviously, we want to describe the relationships among fragments in a more precise way.
Given two fragments $\mathbf L^\circ$ and $\mathbf L^{\circ'}$ (where $\circ,\circ'$ represent generic restrictions among those we have discussed above), interpreted in the same class of relational frames $\mathcal{C}$, we say that $\mathbf L^{\circ}$ and $\mathbf L^{\circ'}$ are {\em model-preserving equivalent}, denoted $\mathbf L^{\circ} \mpequiv \mathbf \mathbf L^{\circ'}$, if there exists an effective translation $(\cdot)'$ from $\mathbf L^{\circ'}$ to $\mathbf L^{\circ'}$ within the same propositional alphabet such that for every model $M$ in $\mathcal{C}$, world $w$ in $M$, and formula $\varphi$ of $\mathbf L^{\circ}$, we have $M, w \mmodels \varphi$ if and only if  $M, w \mmodels \varphi'$, and the other way around.
If there exists a translation from $\mathbf L^{\circ}$ to $\mathbf L^{\circ'}$, but not the other way around, then we say that $\mathbf L^{\circ}$ is {\em weakly less expressive than} $\mathbf L^{\circ'}$, and we denote this situation by $\mathbf L^{\circ} \wlexp \mathbf L^{\circ'}$. Moreover, we say that $\mathbf L^{\circ}$ and $\mathbf L^{\circ'}$, interpreted in the same class of relational frames $\mathcal{C}$, are {\em model-extending equivalent}, denoted $\mathbf L^{\circ}\meequiv \mathbf L^{\circ'}$, if there exists an effective translation $(\cdot)'$ that transforms any $\mathbf L^{\circ}$-formula $\varphi$ written in the alphabet $\mathcal P$ into a $\mathbf L^{\circ'}$-formula written in a suitable alphabet $\mathcal P'\supseteq\mathcal P$, such that for every model $M$ in $\mathcal{C}$ and world $w$ in $M$, we have that $M, w \mmodels \varphi$ if and only if we can extend the valuation $V$ of $M$ from $\mathcal P$ to $\mathcal P'$ to obtain an extended model $M'$ such that $M', w \mmodels \varphi'$, and the other way around. In this case, we refer to $M'$ as an {\em extension} of $M$. Once again, if there exists such a translation from $\mathbf L^{\circ}$ to $\mathbf L^{\circ'}$, but not from $\mathbf L^{\circ'}$ to $\mathbf L^{\circ}$, then $\mathbf L^{\circ}$ is {\em strongly less expressive than} $\mathbf L^{\circ'}$, and we denote this situation by $\mathbf L^{\circ}\slexp \mathbf L^{\circ'}$.

\medskip

So, for example, Theorem~\ref{boolvsclause} proves that modal logic is model-extending equivalent to its clausal form, but not necessarily model-preserving equivalent. In the same way, it is well-known that, in propositional logic, the conjunctive normal form (CNF) is model-extending equivalent, but not necessarily model-preserving equivalent, to the conjunctive normal form where every clause has at most three literals (3CNF).
Thus, the notion of model-extending equivalence is somehow more natural than the notion of model-preserving equivalence; nevertheless, both contribute to the understanding of the relative expressive power of logics.
Notice that if $\mathbf L^{\circ}$ and $\mathbf L^{\circ'}$ are not model-extending equivalent, then they are not model-preserving equivalent, either, and that, consequently, $\mathbf L^{\circ}$ may be weakly less expressive than $\mathbf L^{\circ'}$ but not strongly less expressive. Thus, $\mathbf L^{\circ}\slexp\mathbf L^{\circ'}$ is a stronger notion than $\mathbf L^{\circ}\wlexp\mathbf L^{\circ'}$: in the former case, there is no set of propositional letters that we may add to the language to translate formulas of $\mathbf L^{\circ'}$ into formulas of $\mathbf L^{\circ}$, while in the latter case, we know that  $\mathbf L^{\circ'}$-formulas cannot be translated into $\mathbf L^{\circ}$-formulas within the same propositional alphabet. Given $\mathbf L^{\circ}$ and $\mathbf L^{\circ'}$ such that $\mathbf L^{\circ}$ is a syntactical fragment of $\mathbf L^{\circ'}$, in order to prove that $\mathbf L^{\circ}\wlexp\mathbf L^{\circ'}$ (or that $\mathbf L^{\circ}\slexp\mathbf L^{\circ'}$) we show a $\mathbf L^{\circ'}$-formula $\psi$ that cannot be written in the language $\mathbf L^{\circ}$ without adding new propositional letters (or by adding finitely many new propositional letters). To this end we proceed by contradiction, assuming that a translation $\varphi\in\mathbf L^{\circ}$ does exist, and by building a model for $\psi$ that is not (cannot be extended to) a model of $\varphi$, following three different strategies: we modify the labeling (Theorem~\ref{th:horn_vs_bool} and Theorem~\ref{th:krom_vs_bool}), we modify the structure (Theorem~\ref{th:krombox_vs_krom} and Theorem~\ref{th:kromdia_vs_boxdia}), or we exploit a property of $\mathbf L^{\circ}$ that $\mathbf L^{\circ'}$ does not possess (Theorem~\ref{th:horn_vs_bool}, Theorem~\ref{th:hornbox_vs_boxdia} and Theorem~\ref{th:horndia_vs_boxdia}). All our results, both about the relative expressive power of fragments as well as about the complexity of their satisfiability problem can be immediately extended to the multi-modal version of the languages.

\medskip

Following~\cite{Nguyen2000}, consider any two models $M_1,M_2$ based, respectively, on the sets of worlds $W_1,W_2$, and two worlds $w_1\in W_1$ and $w_2\in W_2$. We say that $M_1$ is {\em less than or equal to} $M_2$ {\em relatively to} $w_1$ {\em and} $w_2$ if and only if there exists a relation $\mathcal S\subseteq W_1 \times W_2$ that respects the following properties: \begin{inparaenum}
	\item\label{prop:less1} $(w_1, w_2) \in \mathcal S$;
	\item\label{prop:less2} for every $(u_1,u_2) \in \mathcal S$, $V_1(u_1) \subseteq V_2(u_2)$;
	\item\label{prop:less3} for every $u_1, v_1 \in W_1$ and $u_2 \in W_2$, if $(u_1, u_2) \in \mathcal S$ and $u_1 \Rx_1 v_1$ then there exists $v_2 \in W_2$ such that $(v_1, v_2) \in \mathcal S$ and $u_2 \Rx_2 v_2$;
	\item\label{prop:less4} for every $u_1\in W_1$, $u_2,v_2 \in W_2$, if $(u_1, u_2) \in \mathcal S$ and $u_2 \Rx_2 v_2$ then there exists $v_1 \in W_1$ such that $(v_1, v_2) \in \mathcal S$ and $u_1 \Rx_1 v_1$.	
\end{inparaenum} We denote this situation by $M_1\leq_{w_1}^{w_2} M_2$. We say that $M_1$ is {\em less than or equal to} $M_2$ ($M_1\leq M_2$) if and  only if there exist $w_1$ and $w_2$ such that $M_1\leq_{w_1}^{w_2} M_2$.

\begin{citedlem}[\cite{Nguyen2000}]\label{lem:least_model}
Let $\varphi$ be a satisfiable formula of \LmonoHorn, where $\Lmono\in\{\Tmono, \SFmono\}$. Then there exists a {\em least model} $M^*$ such that:

\medskip

\begin{compactenum}
\item\label{prop:mstar_is_model} $M^*, w^* \mmodels \varphi$ for some world $w^*$, and
\item\label{prop:mstar_leq_model} for every model $M$ such that $M,w\mmodels\varphi$ for some $w$, it is the case that $M^*\leq_{w^*}^w M$, that is, for every model $M$ of $\varphi$ it is the case that $M^*\leq M$.
\end{compactenum}
\end{citedlem}

\begin{thm}\label{th:horn_vs_bool}
For $\Lmono\in\{\Kmono,\Tmono, \Fmono, \SFmono\}$:
\medskip
\begin{compactenum}
\item  $\LmonoHorn \slexp \LmonoBool$;
\item $\LmonoCore \slexp \LmonoKrom$.
\end{compactenum}
\end{thm}

\begin{proof}
Since $\LmonoHorn$ (resp., \LmonoCore) is a syntactical fragment of $\LmonoBool$ (resp., \LmonoKrom), it only has to be proved that there exists a formula that belongs to $\LmonoBool$ (resp., \LmonoKrom) and that cannot be translated into $\LmonoHorn$ (resp., $\LmonoCore$) over any finite extension of the propositional alphabet. Here, we prove that this is the case for a \LmonoKrom-formula (which is a \LmonoBool-formula as well) that cannot be translated into  \LmonoHorn (and, therefore, to \LmonoCore, either). Let $\mathcal{P} = \{p,q\}$, consider the \LmonoKrom-formula
$$\psi = p \lor q,$$

\noindent and suppose, by contradiction, that there exists  a \LmonoHorn-formula $\varphi$ on some propositional alphabet $\mathcal P' \supseteq \mathcal P$ such that for every model $M$  over the propositional alphabet $\mathcal{P}$, and every world $w$, we have that $M, w\mmodels \psi$ if and only if there exists $M^{\mathcal P'}$ such that $M^{\mathcal P'}, w \mmodels \varphi.$
Let $M_1,M_2$ be two models such that $W_1= W_2=\{w\}$, $V_1(w) = \{p\}$, and $V_2(w) = \{q\}$.
When $\Lmono \in\{\Kmono, \Fmono\}$ then $R_1 = R_2 = \emptyset$, while when $\Lmono \in\{\Tmono, \SFmono\}$ then $R_1 = R_2 = \{(w,w)\}$.
Clearly, $M_1,w\mmodels\psi$ and $M_2,w\mmodels\psi$. Hence, there must exist two extensions $M_1^{\mathcal P'}$ and $M_2^{\mathcal P'}$ such that $M_1^{\mathcal P'}, w \mmodels \varphi$ and $M_2^{\mathcal P'}, w \mmodels \varphi$.
We now distinguish between two cases.

If $\Lmono \in\{\Kmono, \Fmono\}$, consider now model $M^{\mathcal P'}$ obtain by intersecting the valuation of $M_1^{\mathcal P'}$ and $M_2^{\mathcal P'}$, i.e., such that $W = W_1 = W_2 = \{w\}$, $R = \emptyset$ and $V(w) = V_1^{\mathcal P'}(w) \cap V_2^{\mathcal P'}(w)$\footnote{In the following we will generalize this definition of  \emph{intersection of models}, and study its properties w.r.t.\ the Horn box fragments. Such properties do not hold for the full Horn fragment and thus cannot be used in this proof.}.
By definition we have that neither $p$ nor $q$ can belong to $V(w)$, and thus that $M, w \not\mmodels p\vee q$. Since, by hypothesis, $\varphi$ is a translation of $p\lor q$, we have that $M, w \not\mmodels\varphi$ either. Hence, there must exists a clause $\varphi_i = \nabla(\lambda_1 \land \ldots \land \lambda_n \rightarrow \lambda)$ of $\varphi$ such that $M, w \not\mmodels\varphi_i$.
This implies that $\nabla = \boxx^0$ (otherwise the clause would be trivially satisfied), and thus
that $M, w \mmodels\lambda_1 \land \ldots \land \lambda_n$ but $M, w\not\mmodels\lambda$.
Consider now any positive literal $\lambda_j$ in the body of the clause. Three cases may arise:

\medskip

\begin{compactitem}
	\item $\lambda_j = r$ for some $r \in {\mathcal P'}$. Since  $V(w) = V_1^{\mathcal P'}(w) \cap V_2^{\mathcal P'}(w)$ then we have that $r \in  V_1^{\mathcal P'}(w)$ and $r \in V_2^{\mathcal P'}(w)$, from which we can conclude that $M_1^{\mathcal P'}, w \mmodels \lambda_j$ and $M_2^{\mathcal P'}, w \mmodels \lambda_j$;
	\item $\lambda_j = \boxx\lambda'$. Since $w$ has successors neither in $M_1^{\mathcal P'}$ nor in $M_2^{\mathcal P'}$, we have that $\lambda_j$ is trivially satisfied in both models;

	\item $\lambda_j = \diax\lambda'$. Since $w$ has no successors in $M$, then $\lambda_j$ cannot be true in $w$, against the hypothesis that $w$ satisfy the body of the clause.
\end{compactitem}

\medskip

\noindent From the above argument we can conclude that $M_1^{\mathcal P'}, w \mmodels\lambda_1 \land \ldots \land \lambda_n$ and $M_2^{\mathcal P'}, w \mmodels\lambda_1 \land \ldots \land \lambda_n$, from which we have that $M_1^{\mathcal P'}, w \mmodels\lambda$ and $M_2^{\mathcal P'}, w \mmodels\lambda$ (since both $M_1^{\mathcal P'}$ and $M_2^{\mathcal P'}$ satisfy $\varphi$ on $w$). Since $\lambda$ cannot be $\bot$, three cases may arise:

\medskip

\begin{compactitem}
	\item $\lambda = r$ for some $r \in {\mathcal P'}$. Since  $V(w) = V_1^{\mathcal P'}(w) \cap V_2^{\mathcal P'}(w)$ then we have that $r \in  V(w)$, from which we can conclude that $M, w \mmodels \lambda$, against the hypothesis that $M, w\not\mmodels\lambda$;

	\item $\lambda = \boxx\lambda'$. Since $w$ has no successors in $M$, then $\lambda$ is trivially satisfied in $M$, against the hypothesis that $M, w\not\mmodels\lambda$;

	\item $\lambda = \diax\lambda'$. Since $w$ has successors neither in $M_1^{\mathcal P'}$ nor in $M_2^{\mathcal P'}$, we have that $\lambda$ is false in both models, against the hypothesis that $M_1^{\mathcal P'}, w \mmodels\lambda$ and $M_2^{\mathcal P'}, w \mmodels\lambda$.
\end{compactitem}

\medskip

\noindent In all cases a contradiction is found and therefore, $\varphi$ cannot exist.

Now, if $\Lmono \in\{\Tmono, \SFmono\}$ we have that, by Lemma~\ref{lem:least_model}, there exists a least model $M^*$ such that, for some $w^*$, $M^*,w^*\mmodels\varphi$, $M^*\leq_{w^*}^{w_1} M_1$, and $M^*\leq_{w^*}^{w} M_2$. By definition, $V^*(w^*) \subseteq V_1(w)$ and $V^*(w^*) \subseteq V_2(w)$; but this implies that neither $p$ nor $q$ can belong to $V^*(w^*)$, and thus that $M^*, w^* \not\mmodels p\vee q$, which is in contradiction with our hypothesis that $\varphi$ is a translation of $\psi$. Therefore, also in this case $\varphi$ cannot exist.

To conclude, we have shown that for $\Lmono\in\{\Kmono,\Tmono, \Fmono, \SFmono\}$, we cannot express $\psi$ in \LmonoHorn within any finite extension of the propositional alphabet.
\end{proof}

\medskip


\begin{thm}\label{th:krom_vs_bool}
For $\Lmono\in\{\Kmono,\Tmono, \Fmono, \SFmono\}$:
\medskip
\begin{compactenum}
	\item  $\LmonoKrom \wlexp\LmonoBool$;
	\item $\LmonoCore \wlexp \LmonoHorn$.
\end{compactenum}
\end{thm}

\begin{proof}
Since $\LmonoKrom$ (resp., \LmonoCore) is a syntactical fragment of $\LmonoBool$ (resp., \LmonoHorn), it only has to be proved that there exists a formula that belongs to $\LmonoBool$ (resp., \LmonoHorn) and that cannot be translated into $\LmonoKrom$ (resp., $\LmonoCore$) within the same propositional alphabet. Here, we prove that this is the case for a \LmonoHorn-formula (which is a \LmonoBool-formula as well) that cannot be translated into  \LmonoKrom (and, therefore, to \LmonoCore, either). Now, consider the \LmonoHorn-formula
$$\psi = p \land q \rightarrow r,$$

\noindent and suppose, by contradiction, that there exists a \LmonoKrom-formula $\varphi$, written in the propositional alphabet $\{p,q,r\}$, such that for every model $M$ and every world $w$ we have that $M, w\mmodels \psi$ if and only if $M,w \mmodels \varphi.$ We can assume that $\varphi = \varphi_1 \land \ldots \land \varphi_l$ is a conjunction of clauses. Let us denote by $P(\varphi_i)$ the set of propositional letters that occur in $\varphi_i$. Now, consider a model $M = \langle \mathcal F, V\rangle$, where $\mathcal F$ is based on a set of worlds $W$ and it respects the semantical properties of $\Lmono$,
and let $w\in W$ be a world such that $M,w\not\mmodels\psi$. Such a model must exist since $\psi$ is not a tautology.  By hypothesis, $M,w \not\mmodels \varphi$, and since $\varphi$ is a conjunction of Krom clauses,  there must exist at least one clause $\varphi_i = \nabla(\lambda_1 \vee \lambda_2)$ such that $M, w \not\mmodels \varphi_i$. Hence, there must exist a world $w'$ such that $M, w' \not\mmodels (\lambda_1 \vee \lambda_2)$. At this point, three cases may arise (since we are in a fixed propositional alphabet, and we deal with clauses at most binary):

\medskip

\begin{compactitem}

\item $P(\varphi_i) \subseteq \{p,q\}$. In this case, we can build a new model $M' = \langle \mathcal F,V'\rangle$ such that:

 $$V'(p) = V(p),\ V'(q) = V(q),~\mbox{and}~V'(r) = W.$$

 \noindent Since $r$ holds on every world of the model, we have that $M'$ satisfies $\psi$ everywhere, and in particular on $w$. However, since the valuation of $p$ and $q$ are the same of $M$, and since the relational structure has not changed, we have that $M',w' \not\models \lambda_1 \lor \lambda_2$, from which we can conclude that $M',w \not\mmodels \nabla ( \lambda_1 \lor \lambda_2)$ and thus that $w$ does not satisfy $\varphi$.
	
\item $P(\varphi_i) \subseteq \{p,r\}$. In this case, we can build a new model $M' = \langle \mathcal F,V'\rangle$ such that:
$$V'(p) = V(p),\ V'(r) = V(r),~\mbox{and}~V'(q) = \emptyset.$$

\noindent Since $q$ is false on every world of the model, we have that $M'$ satisfies $\psi$ everywhere, and in particular on $w$. However, since the valuation of $p$ and $r$ are the same of $M$, and since the relational structure has not changed, we have that $M',w' \not\models \lambda_1 \lor \lambda_2$, from which we can conclude that $M',w \not\mmodels \nabla ( \lambda_1 \lor \lambda_2)$ and thus that $w$ does not satisfy $\varphi$.

\item $P(\varphi_i) \subseteq \{q,r\}$. In this case, the same argument as above works by using $q$ instead of $p$.

\end{compactitem}

\medskip

\noindent Therefore, $\varphi$ cannot exist, and this means that $\psi$ cannot be expressed in \LmonoKrom within the same propositional alphabet. Since the proof does not depend on the properties of the frame, the claim is proved for any logic $\Lmono\in\{\Kmono,\Tmono, \Fmono, \SFmono\}$.
\end{proof}

\medskip

\begin{cor}\label{th:krom_vs_horn}
For $\Lmono\in\{\Kmono,\Tmono, \Fmono, \SFmono\}$, \LmonoHorn and \LmonoKrom are expressively incomparable within the same propositional alphabet.
\end{cor}

\begin{proof}
As we have seen in Theorem~\ref{th:horn_vs_bool}, the \LmonoKrom-formula $p \lor q$ cannot be translated into \LmonoHorn over any extension of the propositional alphabet, and, as we have seen in Theorem~\ref{th:krom_vs_bool}, the \LmonoHorn-formula $p \land q \rightarrow r$ cannot be translated into \LmonoKrom within the same propositional alphabet.
These two observations, together, prove that we cannot compare \LmonoHorn and \LmonoKrom\ within the same propositional alphabet.
\end{proof}

\medskip

Now, we turn our attention to the relative expressive power for box and diamond fragments, starting with sub-Horn fragments without diamonds. To this end, consider two models $M_1,M_2$ such that all $M_i=(\mathcal F,V_i)$ are based on the same relational frame: we define the {\em intersection} model as the unique model $M_{M_1\cap M_2}=(\mathcal F,V_{V_1\cap V_2})$, where, for each $w\in W$, $V_{V_1\cap V_2}(w)=V_1(w)\cap V_2(w)$. We prove now that the Horn fragments without diamonds are closed under intersection, that is, if two models satisfy a formula then their intersection also satisfies the formula. 

\begin{lem}\label{lem:intersection}
For $\Lmono\in\{\Kmono,\Tmono, \Fmono, \SFmono\}$, \LmonoHornBox\ is closed under intersection of models.
\end{lem}

\begin{proof}
Let $\varphi=\varphi_1\wedge\ldots\wedge\varphi_l$ a \LmonoHornBox-formula such that $M_1,w\mmodels\varphi$ and $M_2,w\mmodels\varphi$, where $M_1=(\mathcal F,V_1)$ and $M_2=(\mathcal F,V_2)$; we want to prove that $M_{M_1\cap M_2},w\mmodels\varphi$. Suppose, by way of contradiction, that $M_{M_1\cap M_2},w\not\mmodels\varphi$. Then, there must be some $\varphi_i$ such that $M_{M_1\cap M_2},w\not\mmodels\varphi_i$. As in Theorem~\ref{th:horn_vs_bool}, we can assume that $\varphi_i$ is a clause of the type $\nabla(\lambda_1\wedge\ldots\wedge\lambda_n\rightarrow\lambda)$. This means that $M_{M_1\cap M_2},w'\mmodels\lambda_1\wedge\ldots\wedge\lambda_n$ and $M_{M_1\cap M_2},w'\not\mmodels\lambda$ for some $w'$. We want to prove that, for each $1\le j\le n$, both $M_1$ and $M_2$ satisfy $\lambda_j$ at $w'$. To see this, let $N=md(\lambda_j)$, and:

\medskip

\begin{itemize}
\item If $N=0$, then $\lambda_j=p$ for some propositional letter $p$; but if $M_{M_1\cap M_2},w'\mmodels p$, then $p\in V_1(w')\cap V_2(w')$, which means that  $M_1,w'\mmodels p$ and  $M_2,w'\mmodels p$.
\item If $N>0$, then $\lambda_j=\boxx\lambda'$. Since $M_{M_1\cap M_2},w'\mmodels \boxx\lambda'$, for every $v$ such that $w'\Rx v$ it is the case that $M_{M_1\cap M_2},v\mmodels \lambda'$. Thus, for every $v$ such that $w'\Rx v$, we know by inductive hypothesis that $M_1,v\mmodels\lambda'$ and $M_2,v\mmodels\lambda'$. But this immediately implies that $M_1,w'\mmodels\boxx\lambda'$ and $M_2,w'\mmodels\boxx\lambda'$, which completes the induction. Now, we have that  $M_1,w'\mmodels\lambda_1\wedge\ldots\wedge\lambda_n$ and $M_2,w'\mmodels\lambda_1\wedge\ldots\wedge\lambda_n$; therefore, $M_1,w'\mmodels\lambda$ and $M_2,w'\mmodels\lambda$.
\end{itemize}

\medskip

\noindent A similar inductive argument shows that $M_{M_1\cap M_2},w'\mmodels\lambda$, implying that $M_{M_1\cap M_2},w\mmodels\varphi_i$; but this contradicts  our hypothesis that $M_{M_1\cap M_2},w\not\mmodels\varphi$.
\end{proof}

\medskip

\begin{thm}\label{th:hornbox_vs_boxdia}
For $\Lmono\in\{\Kmono,\Tmono, \Fmono, \SFmono\}$:
\medskip
\begin{compactenum}
\item $\LmonoHornBox\slexp\LmonoHorn$;
\item $\LmonoCoreBox\slexp\LmonoCore$.
\end{compactenum}
\end{thm}

\begin{proof}
Since $\LmonoHornBox$ (resp., \LmonoCoreBox) is a syntactical fragment of $\LmonoHorn$ (resp., \LmonoCore), it only has to be proved that there exists a formula that belongs to $\LmonoHorn$ (resp., \LmonoCore) and that cannot be translated into $\LmonoHornBox$ (resp., $\LmonoCoreBox$) over any finite extension of the propositional alphabet. Here, we prove that this is the case for a \LmonoCore-formula (which is a \LmonoHorn-formula as well) that cannot be translated into  \LmonoHornBox (and, therefore, to \LmonoCoreBox, either). Let $\mathcal P=\{p\}$, consider the \LmonoHorn-formula
$$\psi = \diax p,$$

\noindent and suppose by contradiction that there exists a propositional alphabet $\mathcal P'\supseteq \mathcal P$ and a \LmonoHornBox formula $\varphi$ written over $\mathcal P'$ such that for every model $M$ over the propositional alphabet $\mathcal P$ and every world $w$ we have that $M,w\mmodels \psi$ if and only if there exists $M^{\mathcal P'}$ such that $M^{\mathcal P'},w\mmodels\varphi.$ Let $M_1=(\mathcal F,V_1)$ and $M_2=(\mathcal F,V_2)$, where $\mathcal F$ is based on the set $W=\{w_0,w_1,w_2\}$. Let $\Rx=\{(w_0,w_1), (w_0,w_2)\}$ (or its reflexive closure in the cases $\Lmono\in\{\Tmono,\SFmono\}$)
, and define the valuation functions $V_1,V_2$ as follows:
$$
V_i(w_j)=\left\{\begin{array}{ll}
\{p\} & \mbox{if }i=j,\\
\emptyset & \mbox{otherwise}.
\end{array}\right.%
$$

\noindent Clearly, $M_1,w_0\mmodels\psi$ and $M_2,w_0\mmodels\psi$; since $\varphi$ is a \LmonoHornBox-translation of $\psi$, it must be the case that, for some extensions $M_1^{\mathcal P'}$ and $M_2^{\mathcal P'}$, we have that $M_1^{\mathcal P'},w_0\mmodels\varphi$ and $M_2^{\mathcal P'},w_0\mmodels\varphi$. By Lemma~\ref{lem:intersection}, their intersection model $M_{M_1^{\mathcal P'}\cap M_2^{\mathcal P'}}$ is such that $M_{M_1^{\mathcal P'}\cap M_2^{\mathcal P'}},w_0\mmodels\varphi$. But $p \not\in V_{V_1^{\mathcal P'}\cap V_2^{\mathcal P'}}(w)$ for every $w\in W$, so $M_{M_1^{\mathcal P'}\cap M_2^{\mathcal P'}},w\not\mmodels\psi$. This contradicts the hypothesis that $\varphi$ is a translation of $\psi$.
\end{proof}

\medskip

To establish the expressive power of \LmonoHornDia and \LmonoCoreDia with respect to other fragments, we now prove a closure property similar to Lemma~\ref{lem:intersection}. Consider two models $M_1 = (\mathcal F_1, V_1)$, $M_2 = (\mathcal F_2, V_2)$ based on two (possibly different) relational frames. We define the {\em product} model as the unique model $M_{M_1\times M_2}=(\mathcal F_{\mathcal F_1 \times \mathcal F_2},V_{V_1\times V_2})$, where: \begin{inparaenum}\item $\mathcal F_{\mathcal F_1 \times \mathcal F_2} = (W_1 \times W_2, \Rprod)$, that is, worlds are all and only the pairs of worlds from $W_1$ and $W_2$;  $(w_1,w_2)\Rprod(w_1',w_2')$ if and only if $w_1 R_1  w_1'$ and $w_2 R_2 w_2'$, that is, worlds in  $\mathcal F_{\mathcal F_1 \times \mathcal F_2}$ are connected to each other as the component worlds were connected in $\mathcal F_1$ and $\mathcal F_2$; and \item $V_{V_1\times V_2}((w_1,w_2)) = V_1(w_1) \cap V_2(w_2)$. \end{inparaenum} We prove now that the Horn fragments without boxes are closed under product.

\begin{lem}\label{lem:product}
For $\Lmono\in\{\Kmono,\Tmono, \Fmono, \SFmono\}$, \LmonoHornDia\ is closed under product of models.
\end{lem}

\begin{proof}
Let $\varphi=\varphi_1\wedge\ldots\wedge\varphi_l$ be a \LmonoHornDia-formula such that $M_1,w_1\mmodels\varphi$ and $M_2,w_2\mmodels\varphi$. We want to prove that $M_{M_1\times M_2},(w_1,w_2)\mmodels\varphi$; suppose by way of contradiction, that $M_{M_1\times M_2},(w_1,w_2)\not\mmodels\varphi$. Then, there must be some $\varphi_i$ such that $M_{M_1\times M_2},(w_1,w_2)\not\mmodels\varphi_i$. As in Theorem~\ref{th:horn_vs_bool}, we can assume that $\varphi_i$ is a clause of the type $\nabla(\lambda_1\wedge\ldots\wedge\lambda_n\rightarrow\lambda)$. This means that $M_{M_1\times M_2},(w_1',w_2')\mmodels\lambda_1\wedge\ldots\wedge\lambda_n$ and $M_{M_1\times M_2},(w_1',w_2')\not\mmodels\lambda$ for some $(w_1',w_2')$. We want to prove that, for each $1\le j\le n$,  $M_1$ and $M_2$ satisfy $\lambda_j$ at, respectively, $w_1'$ and $w_2'$. To see this, let $N=md(\lambda_j)$, and:

\medskip

\begin{itemize}
\item If $N=0$, then $\lambda_j=p$ for some propositional letter $p$: by the definition of product, we have that $M_{M_1\times M_2},(w_1',w_2')\mmodels p$ if and only if $p\in V_1(w_1')\cap V_2(w_2')$, which means that  $M_1,w_1'\mmodels p$ and  $M_2,w_1'\mmodels p$.
\item If $N>0$, then $\lambda_j=\diax\lambda'$. Since $M_{M_1\times M_2},(w_1',w_2')\mmodels \diax\lambda'$, then there exists $(v_1,v_2)$ such that $(w_1',w_2')\Rprod (v_1,v_2)$ and $M_{M_1\times M_2},(v_1,v_2)\mmodels \lambda'$. We know by inductive hypothesis that $M_1,v_1\mmodels\lambda'$ and $M_2,v_2\mmodels\lambda'$ and that, by definition of product, $w_1'R_1 v_1$ and $w_2' R_2 v_2$. But this immediately implies that $M_1,w_1'\mmodels\diax\lambda'$ and $M_2,w_2'\mmodels\diax\lambda'$, which completes the induction. Now, we have that  $M_1,w_1'\mmodels\lambda_1\wedge\ldots\wedge\lambda_n$ and $M_2,w_2'\mmodels\lambda_1\wedge\ldots\wedge\lambda_n$; therefore, $M_1,w_1'\mmodels\lambda$ and $M_2,w_2'\mmodels\lambda$.
\end{itemize}

\medskip

\noindent A similar argument proves that $M_{M_1\times M_2},(w_1',w_2')\mmodels\lambda$, implying that $M_{M_1\times M_2},(w_1,w_2)\mmodels\varphi_i$, in contradiction with  the hypothesis that $M_{M_1\times M_2},(w_1,w_2)\not\mmodels\varphi$.
\end{proof}

\medskip

\begin{thm}\label{th:horndia_vs_boxdia}
For $\Lmono\in\{\Kmono, \Fmono\}$:
\medskip
\begin{compactenum}
\item $\LmonoHornDia\slexp\LmonoHorn$;
\item $\LmonoCoreDia\slexp\LmonoCore$.
\end{compactenum}
\end{thm}

\begin{proof}
Since $\LmonoHornDia$ (resp., \LmonoCoreDia) is a syntactical fragment of $\LmonoHorn$ (resp., \LmonoCore), we know that $\LmonoHornDia$ can be translated into $\LmonoHorn$ and that $\LmonoCoreDia$ can be translated into $\LmonoCore$. It remains to be proved that there exists a formula that belongs to $\LmonoHorn$ (resp., \LmonoCore) and that cannot be translated into $\LmonoHornDia$ (resp., $\LmonoCoreDia$) over any finite extension of the propositional alphabet. Here, we prove that this is the case for a \LmonoCore-formula (which is a \LmonoHorn-formula as well) that cannot be translated into  \LmonoHornDia (and, therefore, to \LmonoCoreDia, either). Let $\mathcal P=\{p, q\}$, consider the \LmonoCore-formula
$$\psi = \boxx p \rightarrow q,$$

\noindent and suppose by contradiction that there exists a propositional alphabet $\mathcal P'\supseteq \mathcal P$ and a \LmonoHornDia formula $\varphi$ written over $\mathcal P'$ such that for every model $M$ over the propositional alphabet $\mathcal P$ and every world $w$ we have that $M,w\mmodels \psi$ if and only if there exists $M^{\mathcal P'}$ such that $M^{\mathcal P'},w\mmodels\varphi.$ Let $M_1=(\mathcal F_1,V_1)$ and $M_2=(\mathcal F_2,V_2)$, where $\mathcal F_1$ is based on the set $W=\{w_0,w_1\}$ and  the accessibility relation $\Rx_1=\{(w_0,w_1)\}$, while $\mathcal F_2$ is based on $\{v_0\}$ and  $\Rx_2 = \emptyset$. Define the valuation function $V_1$ as always empty, and let $q\in V_2(v_0)$. Clearly, $M_1,w_0\mmodels\psi$ and $M_2,v_0\mmodels\psi$. Since $\varphi$ is a \LmonoHornDia-translation of $\psi$, it must be the case that, for some extensions $M_1^{\mathcal P'}$ and $M_2^{\mathcal P'}$, we have that $M_1^{\mathcal P'},w_0\mmodels\varphi$ and $M_2^{\mathcal P'},v_0\mmodels\varphi$. By Lemma~\ref{lem:product}, their product model $M_{M_1^{\mathcal P'}\times M_2^{\mathcal P'}}$ is such that $M_{M_1^{\mathcal P'}\times M_2^{\mathcal P'}},(w_0,v_0)\mmodels\varphi$. Notice that $q \not\in V_{V_1^{\mathcal P'}\times V_2^{\mathcal P'}}(w_0,v_0)$ and that $(w_0,v_0)$ has no $\Rx$-successors. Hence, we have that $M_{M_1^{\mathcal P'}\times M_2^{\mathcal P'}},(w_0,v_0)\mmodels\boxx p$ but $M_{M_1^{\mathcal P'}\times M_2^{\mathcal P'}},(w_0,v_0)\not\mmodels q$, in contradiction with the hypothesis that $\varphi$ is a translation of $\psi$. Therefore, $\varphi$ cannot exist, and this means that $\psi$ cannot be expressed in \LmonoHornDia within any finite extension of the propositional alphabet.
\end{proof}

\medskip

\noindent It is worth to observe that restricting the above result to classes of frames not necessarily reflexive is essential. Not only the above counter-example does not work assuming reflexiveness, but it can be easily proved that the argument cannot be fixed. As a matter of fact, both \TmonoHorn\ and \SFmonoHorn\ are closed under product of models, so that this particular characteristics cannot be used to distinguish \LmonoHornDia\ from \LmonoHorn\ when $\Lmono\in\{\Tmono, \SFmono\}$, when new propositional letters are allowed in the translation. Later on, in Theorem~\ref{th:kromdia_vs_boxdia}, we will show that within the same propositional alphabet we can indeed distinguish between \LmonoHornDia\ and \LmonoHorn\ also in reflexive frames.

\begin{fact}\label{fact:product}
\TmonoHorn\ and \SFmonoHorn\ are closed under product of models.
\end{fact}

\begin{proof}
Let $\varphi=\varphi_1\wedge\ldots\wedge\varphi_l$ be a \TmonoHorn-formula such that $M_1,w_1\mmodels\varphi$ and $M_2,w_2\mmodels\varphi$. We want to prove that $M_{M_1\times M_2},(w_1,w_2)\mmodels\varphi$; suppose by way of contradiction, that $M_{M_1\times M_2},(w_1,w_2)\not\mmodels\varphi$. Then, there must be some $\varphi_i$ such that $M_{M_1\times M_2},(w_1,w_2)\not\mmodels\varphi_i$. As in Theorem~\ref{th:horn_vs_bool}, we can assume that $\varphi_i$ is a clause of the type $\nabla(\lambda_1\wedge\ldots\wedge\lambda_n\rightarrow\lambda)$. This means that $M_{M_1\times M_2},(w_1',w_2')\mmodels\lambda_1\wedge\ldots\wedge\lambda_n$ and $M_{M_1\times M_2},(w_1',w_2')\not\mmodels\lambda$ for some $(w_1',w_2')$. We want to prove that, for each $1\le j\le n$,  $M_1$ and $M_2$ satisfy $\lambda_j$ at, respectively, $w_1'$ and $w_2'$. To see this, let $N=md(\lambda_j)$, and:

\medskip

\begin{itemize}
\item  If $N=0$, then $\lambda_j=p$ for some propositional letter $p$: by the definition of product, we have that $M_{M_1\times M_2},(w_1',w_2')\mmodels p$ if and only if $p\in V_1(w_1')\cap V_2(w_2')$, which means that  $M_1,w_1'\mmodels p$ and  $M_2,w_1'\mmodels p$.
\item If $N>0$, then $\lambda_j=\boxx\lambda'$ or $\lambda_j=\diax\lambda'$; the latter case has been dealt with in Lemma~\ref{lem:product}, so we can focus on the former one. Since $M_{M_1\times M_2},(w_1',w_2')\mmodels \boxx\lambda'$, then for all $(v_1,v_2)$ such that $(w_1',w_2')\Rprod (v_1,v_2)$, we have that $M_{M_1\times M_2},(v_1,v_2)\mmodels \lambda'$. Let us prove that $M_1,w_1'\mmodels\boxx\lambda'$ and $M_2,w_2'\mmodels\boxx\lambda'$. Consider any $v_1$ such that $w'_1 R_1 v_1$; since $R_2$ is reflexive, we have that $w'_2 R_2 w'_2$, hence, by definition of product, $(w'_1, w'_2) \Rprod (v_1, w'_2)$, which implies that $M_{M_1\times M_2},(v_1,w'_2)\mmodels \lambda'$. We know by inductive hypothesis that $M_1,v_1\mmodels\lambda'$ and $M_2,w'_2\mmodels\lambda'$. But this immediately implies that $M_1,w_1'\mmodels\boxx\lambda'$. Similarly, we can prove that $M_1,w_2'\mmodels\boxx\lambda'$, which completes the induction. Now, we have that  $M_1,w_1'\mmodels\lambda_1\wedge\ldots\wedge\lambda_n$ and $M_2,w_2'\mmodels\lambda_1\wedge\ldots\wedge\lambda_n$; therefore, $M_1,w_1'\mmodels\lambda$ and $M_2,w_2'\mmodels\lambda$.
 \end{itemize}

\medskip

\noindent A similar argument proves that $M_{M_1\times M_2},(w_1',w_2')\mmodels\lambda$, implying that $M_{M_1\times M_2},(w_1,w_2)\mmodels\varphi_i$, in contradiction with  the hypothesis that $M_{M_1\times M_2},(w_1,w_2)\not\mmodels\varphi$.
\end{proof}

\medskip

The argument of Theorem~\ref{th:hornbox_vs_boxdia}, based on the intersection of models, cannot be replicated to establish the relationship between sub-Krom fragments; it turns out that in this case the possibility of expanding the propositional alphabet does make the difference, as the following result shows.

\begin{thm}\label{th:krombox_vs_krom}
For $\Lmono\in\{\Kmono,\Tmono, \Fmono, \SFmono\}$:
\medskip
\begin{compactenum}
\item\label{KromBoxVsKrom} $\LmonoKromBox\meequiv\LmonoKrom$;
\item $\LmonoKromBox\wlexp\LmonoKrom$.
\end{compactenum}
\end{thm}

\begin{proof}
The first result is easy to prove. Suppose that
$$\varphi=\nabla_1(\lambda_1^1\vee\lambda_2^1)\wedge \nabla_2(\lambda_1^2\vee\lambda_2^2)\wedge\ldots\wedge\nabla_i(\lambda_1^i\vee\lambda_2^i)\wedge\ldots
\wedge\nabla_l(\lambda_1^l\vee\lambda_2^l)$$

\noindent is a $\LmonoKrom$-formula, where, as in Theorem~\ref{th:krom_vs_bool}, we treat literals as special clauses. There are two cases. First, suppose that $\lambda_1^i=\diax\lambda$,  for some $1\leq i\leq l$, where $\lambda$ is a positive literal. We claim that the \LmonoKromBox-formula
$$\varphi'=\nabla_1(\lambda_1^1\vee\lambda_2^1)\wedge \nabla_2(\lambda_1^2\vee\lambda_2^2)\wedge\ldots\wedge\nabla_i(\neg\boxx p\vee\lambda_2^i)\wedge\nabla_i\boxx(p\vee\lambda)\wedge\ldots
\wedge\nabla_l(\lambda_1^l\vee\lambda_2^l),$$

\noindent where $p$ is a fresh propositional variable, is model-extending equivalent to $\varphi$. To see this, let $\mathcal P$ the propositional alphabet in which $\varphi$ is written, and let $\mathcal P'=\mathcal P\cup\{p\}$, and consider a model $M=(\mathcal F,V)$ such that, for some world $w$, it is the case that $M,w\mmodels\varphi$; in particular, it is the case that $M,w\mmodels\nabla_i(\lambda_1^i\vee\lambda_2^i)$; let $W_i\subseteq W$ be the set of worlds reachable from $w$ via the universal prefix $\nabla_i$, and consider $v\in W_i$. If $M,v\mmodels\lambda_2^i$ we can extend $M$ to a model $M^\mathcal P=(\mathcal F,V^\mathcal P)$ such that it satisfies $p$ on every world reachable from $v$, if any, and both substituting clauses are satisfied. If, on the other hand, $M,v\mmodels\diax\lambda$, for some $t$ such that $v\Rx t$ we have that $M,t\mmodels\lambda$; we can now extend $M$ to a model $M^\mathcal P=(\mathcal F,V^\mathcal P)$ such that it satisfies $\neg p$ on $t$, and $p$ on every other world reachable from $v$, if any, and, again, both substituting clauses are satisfied. A reversed argument proves that if $M,w\mmodels\varphi'$ it must be the case that $M,w\mmodels\varphi$. If, as a second case, $\lambda_1^i=\neg\diax\lambda$, where $\lambda$ is a positive literal, then the translating formula is
$$\varphi'=\nabla_1(\lambda_1^1\vee\lambda_2^1)\wedge \nabla_2(\lambda_1^2\vee\lambda_2^2)\wedge\ldots\wedge\nabla_i(\boxx p\vee\lambda_2^i)\wedge\nabla_i\boxx(\neg p\vee\neg \lambda)\wedge\ldots
\wedge\nabla_l(\lambda_1^l\vee\lambda_2^l),$$

\noindent and the proof of model-extending equivalence is identical to the above one.

\medskip

In order to prove the second result, we observe that since $\LmonoKromBox$ is a syntactical fragment of $\LmonoKrom$ it only has to be proved that there exists a formula that belongs to $\LmonoKrom$ and that cannot be translated into $\LmonoKromBox$ within the same propositional alphabet. Let $\mathcal P=\{p\}$, consider the \LmonoKrom-formula
$$\psi = \diax p,$$

\noindent and suppose by contradiction that there exists a \LmonoKromBox formula $\varphi$ such that for every model $M$ over the propositional alphabet $\mathcal P$ and every world $w$ we have that $M,w\mmodels \psi$ if and only if $M,w\mmodels\varphi.$ Once again, we can safely assume that $\varphi=\varphi_1\wedge\varphi_2\wedge\ldots\wedge\varphi_l$, and that each $\varphi_i$ is a clause. Consider a model $M = \langle \mathcal F, V\rangle$, where $\mathcal F$ is based on a set of worlds $W$, and let $w\in W$ be a world such that $M,w\not\mmodels\psi$. Such a model must exist since $\psi$ is not a tautology. Since $\varphi$ is a conjunction of Krom clauses, we have that there must exist at least one clause $\varphi_i = \nabla(\lambda_1 \vee \lambda_2)$ such that $M, w \not\mmodels \varphi_i$. Hence, there must exist a world $w'$ such that $M, w' \not\mmodels (\lambda_1 \vee \lambda_2)$. Now, consider the model $\overline M$ obtained from $M$ by extending the set of worlds $W$ to $\overline W=W\cup\{\overline w\}$, and the relation $\Rx$ to (the reflexive, transitive, or reflexive and transitive closure, depending on $\Lmono$, of) $\overline\Rx=\Rx\cup\{(w,\overline w)\}$. We define $\overline V(\overline w)=\{p\}$: clearly, $\overline M,w\mmodels\psi$. To prove that $\overline M,w'\not\mmodels\lambda_1\vee\lambda_2$, we first prove the following technical result for positive literals:
$$M,t\mmodels\lambda\Leftrightarrow \overline M,t\mmodels\lambda,$$

\noindent for every $t\in W$ and positive literal $\lambda$. We do so by induction on $N=md(\lambda)$, as follows:

\medskip

\begin{itemize}
\item  If $N=0$, then $\lambda$ is a propositional letter (the cases in which $\lambda=\top$ are trivial): the valuation of $t$ has not changed from $M$ to $\overline M$, and therefore we have the claim immediately.
\item If $N>0$, then assume that $\lambda=\boxx\lambda'$, and $\lambda'$ is a positive literal. By definition, $M,t\mmodels\boxx\lambda'$ if and only if for every $t'$ such that $t\Rx t'$, if any, it is the case that $M,t'\mmodels\lambda'$. Clearly, if $t\neq w$ and $(t,w) \not\in\;\Rx$ then the set of reachable worlds from $t$ has not changed from $M$ to $\overline M$; by inductive hypothesis, for every $t'$, we have that $M,t'\mmodels \lambda'$ if and only if $\overline M,t'\mmodels \lambda'$, and, therefore, $M,t\mmodels\lambda$ if and only if $\overline M,t\mmodels \lambda$, as we wanted. Otherwise, suppose that either $t=w$ or $(t,w)\in\;\Rx$, and recall that $(w, \overline w)\in\overline \Rx$. Notice that the set of reachable worlds from $t$ has changed in the case $t=w$ by definition of $\overline R$, and in the case $(t,w)\in\;\Rx$ when $\Lmono\in\{\Fmono, \SFmono\}$. If $M,t\not\mmodels\boxx\lambda'$, then
there exists some $t'\in W$ such that $t\Rx t'$ and $M,t'\not\mmodels\lambda'$
; so, by inductive hypothesis, $\overline M,t'\not\mmodels\lambda'$, which means that $\overline M,t\not\mmodels\boxx\lambda'$.
If, on the other hand, $M,t\mmodels\boxx\lambda'$, then: \begin{inparaenum} \item if $\lambda'=\top$, then $\overline M,t\mmodels\boxx\top$ independently from the presence of $\overline w$; and \item if $\lambda'=p$, then $\overline M,t\mmodels\boxx\lambda'$ because $p\in \overline V(\overline w)$ by construction; \item if $\lambda'=\boxx \lambda''$, then $\overline M,\overline w\mmodels\lambda'$, because $\overline w$ has no successors.\end{inparaenum}
\end{itemize}

\medskip

\noindent Since by hypothesis $M,w'\not\mmodels\lambda_1\vee\lambda_2$, we have that $M,w'\not\mmodels\lambda_j$ for $j \in \{1,2\}$. By the syntax of $\LmonoKromBox$, $\lambda_j$ can be either a positive or a negative literal.
If $\lambda_j$ is a positive literal, then from the above result we can directly conclude that $\overline M,w'\not\mmodels\lambda_j$. Conversely, if $\lambda_j = \neg\lambda$ is a negative literal, then we have that $M,w'\mmodels\lambda$ which implies, by the above result, that $\overline M,w'\mmodels\lambda$, that is, $\overline M,w'\not\mmodels \lambda_j$.
This implies that $\overline M,w'\not\mmodels\lambda_1\vee\lambda_2$, which means that $\overline M,w\not\mmodels\varphi_i$, that is, $\overline M,w\not\mmodels\varphi$. Therefore $\varphi$ cannot be a translation of $\psi$, and the claim is proved.
\end{proof}

\medskip

The following result deals with sup-propositional fragments without boxes; as before, the argument of Theorem~\ref{th:horndia_vs_boxdia}, based on the product of models, cannot be replicated. We would like to remark that Theorem~\ref{th:horndia_vs_boxdia} gives us a stronger result for \LmonoHornDia and \LmonoCoreDia when $\Lmono\in\{\Kmono, \Fmono\}$.

\begin{thm}\label{th:kromdia_vs_boxdia}
For $\Lmono\in\{\Kmono,\Tmono, \Fmono, \SFmono\}$:
\medskip
\begin{compactenum}
\item\label{KromDiaVsKrom} $\LmonoKromDia\meequiv\LmonoKrom$;
\item $\LmonoKromDia\wlexp\LmonoKrom$;
\item $\LmonoHornDia\wlexp\LmonoHorn$;
\item $\LmonoCoreDia\wlexp\LmonoCore$.
\end{compactenum}
\end{thm}

\begin{proof}
As before, the first result is relatively easy to see. Suppose that
$$\varphi=\nabla_1(\lambda_1^1\vee\lambda_2^1)\wedge \nabla_2(\lambda_1^2\vee\lambda_2^2)\wedge\ldots\wedge\nabla_i(\lambda_1^i\vee\lambda_2^i)\wedge\ldots
\wedge\nabla_l(\lambda_1^l\vee\lambda_2^l)$$

\noindent is a $\LmonoKrom$-formula, where, as in Theorem~\ref{th:krom_vs_bool}, we treat literals as special clauses. There are two cases. Suppose, first, that $\lambda_1^i=\boxx\lambda$, where $\lambda$ is a positive literal. We claim that the \LmonoKromDia-formula
$$\varphi'=\nabla_1(\lambda_1^1\vee\lambda_2^1)\wedge \nabla_2(\lambda_1^2\vee\lambda_2^2)\wedge\ldots\wedge\nabla_i(\neg\diax p\vee\lambda_2^i)\wedge\nabla_i\boxx(p\vee\lambda)\wedge\ldots
\wedge\nabla_l(\lambda_1^l\vee\lambda_2^l),$$

\noindent where $p$ is a fresh propositional variable, is model-extending equivalent to $\varphi$. To see this, let $\mathcal P$ the propositional alphabet in which $\varphi$ is written, and let $\mathcal P'=\mathcal P\cup\{p\}$, and consider a model $M=(\mathcal F,V)$ such that, for some world $w$, it is the case that $M,w\mmodels\varphi$; in particular, it is the case that $M,w\mmodels\nabla_i(\lambda^i_1\vee \lambda^i_2)$; let $W_i\subseteq W$ be the set of worlds reachable from $w$ via the universal prefix $\nabla_i$, and consider $v\in W_i$. If $M,v\mmodels\lambda_2^i$ we can extend $M$ to a model $M^\mathcal P=(\mathcal F,V^\mathcal P)$ such that it satisfies $p$ on every world reachable from $v$, if any, and both substituting clauses are satisfied. If, on the other hand, $M,v\mmodels\boxx\lambda_1^i$, for every $t$ such that $v\Rx t$ we have that $M,t\mmodels\lambda$; we can now extend $M$ to a model $M^\mathcal P=(\mathcal F,V^\mathcal P)$ such that it satisfies $\neg p$ on every such $t$ (if any), and, again, both substituting clauses are satisfied. A reversed argument proves that if $M,w\mmodels\varphi'$ it must be the case that $M,w\mmodels\varphi$. If, as a second case, $\lambda_1^i=\neg\boxx\lambda$, where $\lambda$ is a positive literal, then the translating formula is
$$\varphi'=\nabla_1(\lambda_1^1\vee\lambda_2^1)\wedge \nabla_2(\lambda_1^2\vee\lambda_2^2)\wedge\ldots\wedge\nabla_i(\diax p\vee\lambda_2^i)\wedge\nabla_i\boxx(\neg p\vee\neg \lambda)\wedge\ldots
\wedge\nabla_l(\lambda_1^l\vee\lambda_2^l),$$

\noindent and the proof of model-extending equivalence is identical to the above one.

\medskip

As far as the other three relationships are concerned, since $\LmonoHornDia$ (resp., \LmonoKromDia and \LmonoCoreDia) is a syntactical fragment of $\LmonoHorn$ (resp., \LmonoKrom and \LmonoCore), we only have to prove that there exists a formula that belongs to $\LmonoHorn$ (resp., $\LmonoKrom$, $\LmonoCore$) that cannot be translated into $\LmonoHornDia$ (resp., $\LmonoKromDia$, $\LmonoCoreDia$) within the same propositional alphabet. To this end, we consider a $\LmonoCore$-formula and we prove that it cannot be translated into $\LmonoHornDia$ nor to $\LmonoKromDia$ within the same propositional alphabet, implying that it cannot be translated into $\LmonoCoreDia$, either. Consider the following formula:
$$\psi=\boxx p\rightarrow q.$$

\noindent We prove a very general claim: there is no clausal form formula of the diamond fragment of $\Lmono$ that translates $\psi$ within the propositional alphabet $\{p,q\}$.
Notice that this is not in contradiction with Theorem~\ref{boolvsclause}, since the translation of a generic formula in clausal form may add new propositional variables~\cite{nguyen2004complexity}.
Suppose, by contradiction, that there exists a conjunction $\varphi$ of box-free clauses, such that for every model $M$ over the propositional alphabet $\mathcal P=\{p,q\}$ and every world $w$ we have that $M,w\mmodels \psi$ if and only if $M,w\mmodels\varphi.$ Let $\varphi=\varphi_1 \land \ldots \land \varphi_n$, where each $\varphi_i$ is in its generic form $\nabla(\neg\lambda_1\vee\ldots\vee\neg\lambda_n\vee\lambda_{n+1}\vee\ldots\vee\lambda_{n+m})$. As always, literals are treated as special clauses. Now, consider a model $M = \langle \mathcal F, V\rangle$, where $\mathcal F$ is based on a set of worlds $W$, and let $w\in W$ be a world such that $M,w\not\mmodels\psi$, and such that there exists at least one $v$ such that $w\Rx v$. Since $M,w\not\mmodels\psi$,  we have that $q\notin V(w)$ and for each $v$ such that $w\Rx v$ it is the case that $p\in V(v)$. Since $\varphi$ is a translation of $\psi$, it must be the case that $M,w\not\mmodels\varphi$, which implies that there must be a clause $\varphi_i$ such that $M,w\not\mmodels\varphi_i$, that is, there must be a world $w'$ such that  $M,w'\not\mmodels(\neg\lambda_1\vee\ldots\vee\neg\lambda_n\vee\lambda_{n+1}\vee\ldots\vee\lambda_{n+m})$. Now, consider the model $\overline M$ obtained from $M$ by extending the set of worlds $W$ to $\overline W=W\cup\{\overline w\}$, and the relation $\Rx$ to (the reflexive, transitive, or reflexive and transitive closure, depending on $\Lmono$, of) $\bar\Rx=\Rx\cup\{(w,\overline w)\}$. We set $\overline V(\overline w)=\emptyset$; clearly, $\overline M,w\mmodels \psi$. We want to prove that $\overline M,w'\not\mmodels\varphi_i$. Let us prove the following:
$$M,t\mmodels\lambda\Leftrightarrow \overline{M},t\mmodels\lambda,$$

\noindent for every $t\in W$ and positive literal $\lambda$. We do so by induction on $N=md(\lambda)$.

\medskip

\begin{itemize}
\item If $N=0$, then $\lambda$ is a propositional letter (the cases in which $\lambda=\top$ are trivial): the valuation of $t$ has not changed from $M$ to $\overline M$, and therefore we have the claim immediately.

\item If $N>0$, then we assume that $\lambda=\diax\lambda'$, and $\lambda'$ is a positive literal. By definition, $M,t\mmodels\diax\lambda'$ if and only if there exists some $t'$ such that $t\Rx t'$ and $M,t'\mmodels\lambda'$. Clearly, if $t\neq w$  and $(t,w)\not\in R$  then the set of reachable worlds from $t$ has not changed from $M$ to $\overline M$; by inductive hypothesis, $M,t'\mmodels \lambda'$ if and only if $\overline{M},t'\mmodels \lambda'$, and, therefore, $M,t\mmodels\diax\lambda$ if and only if $\overline M,t\mmodels\diax\lambda$, as we wanted. Otherwise, suppose that either $t=w$ or $(t,w)\in\;\Rx$, and recall that $(w, \overline{w})\in\overline \Rx$. Notice that the set of reachable worlds from $t$ has changed in the case $t=w$ by definition of $\overline R$, and in the case $(t,w)\in\;\Rx$ when $\Lmono\in\{\Fmono, \SFmono\}$. If $M,t\mmodels\diax\lambda'$, then there exists some $t'\in W$ such that $t\Rx t'$ and $M,t'\mmodels\lambda'$
; so, by inductive hypothesis, $\overline M,t'\mmodels\lambda'$, which means that $\overline M,t\mmodels\diax\lambda'$. If, on the other hand, $M,t\not\mmodels\diax\lambda'$, then: \begin{inparaenum} \item $\lambda'\neq\top$, because we have built $M$ in such a way that $w$ has a successor, and \item for every $t'$ such that $t\Rx t'$ it is the case that $M,t'\not\mmodels\lambda'$. \end{inparaenum} Since $\overline V(\overline w)=\emptyset$, and $\lambda'$ is positive and $\overline w$ has no successors, for every $t'$ such that $t	\mathrel{\overline{R}} t'$ it is the case that $\overline M,t'\not\mmodels\lambda'$, and, therefore, $\overline M,t\not\mmodels\diax\lambda'$, as we wanted.
\end{itemize}

\medskip

\noindent This means that $M,w'\not\models\neg\lambda_1\vee\ldots\vee\neg\lambda_n\vee\lambda_{n+1}\vee\ldots\vee\lambda_{n+m}$ implies that  $\overline M,w'\not\models\neg\lambda_1\vee\ldots\vee\neg\lambda_n\vee\lambda_{n+1}\vee\ldots\vee\lambda_{n+m}$, that is, $\overline M,w\not\models\varphi_i$. Hence, $\overline M,w\not\mmodels\varphi$. Therefore, $\varphi$ cannot exist, and this means that $\psi$ cannot be expressed neither in \LmonoKromDia, nor in \LmonoHornDia, nor in \LmonoCoreDia, within the same propositional alphabet.
\end{proof}

\medskip

\begin{cor}\label{cor:box_vs_dia}
The following results hold:
\medskip

\begin{compactenum}
\item\label{cor312:point1} For $\Lmono\in\{\Kmono, \Fmono\}$, \LmonoHornBox and \LmonoHornDia are expressively incomparable;
\item\label{cor312:point2} For $\Lmono\in\{\Kmono, \Fmono\}$, \LmonoCoreBox and \LmonoCoreDia are expressively incomparable;
\item\label{cor312:point3} For $\Lmono\in\{\Kmono,\Tmono, \Fmono, \SFmono\}$, \LmonoKromBox and \LmonoKromDia are expressively incomparable within the same propositional alphabet.
\end{compactenum}
\end{cor}

\begin{proof}
As we have seen in Theorem~\ref{th:hornbox_vs_boxdia}, the \LmonoCoreDia-formula (which is also a \LmonoHornDia-formula) $\diax p$ cannot be translated into \LmonoHornBox (and therefore it cannot be translated into \LmonoCoreBox either), over any finite extension of the propositional alphabet, and, as we have seen in Theorem~\ref{th:horndia_vs_boxdia}, the  \LmonoCoreBox-formula $\boxx p\rightarrow q$ (which is also a \LmonoHornBox-formula) cannot be translated into \LmonoHornDia (and therefore it cannot be translated into \LmonoCoreDia either), over any finite extension of the propositional alphabet when $\Lmono\in\{\Kmono, \Fmono\}$. These two observations, together, show that we cannot compare  \LmonoHornBox with \LmonoHornDia, nor \LmonoCoreBox with \LmonoCoreDia, proving points \ref{cor312:point1} and \ref{cor312:point2}. Similarly, Theorem~\ref{th:krombox_vs_krom} proves that the \LmonoKromDia-formula $\diax p$ cannot be translated into \LmonoKromBox, and Theorem~\ref{th:kromdia_vs_boxdia} proves that the \LmonoKromBox-formula $\boxx p\rightarrow q$ cannot be translated into \LmonoKromDia, all this within the same propositional alphabet; these two observations, together, imply that, within the same propositional alphabet, we cannot compare \LmonoKromBox and \LmonoKromDia either,  proving \ref{cor312:point3}.
\end{proof}

\medskip

\begin{cor}\label{cor:horn_vs_krom}
For $\Lmono\in\{\Kmono,\Tmono, \Fmono, \SFmono\}$, the following results hold:

\medskip

\begin{compactenum}
\item\label{cor313:point1} $\LmonoCoreBox\slexp\LmonoKrom$, $\LmonoKromBox$, $\LmonoKromDia$;
\item\label{cor313:point2} $\LmonoCoreDia\slexp\LmonoKrom$, $\LmonoKromBox$, $\LmonoKromDia$.
\item\label{cor313:point3} $\LmonoHornDia,\LmonoHornBox,\LmonoHorn$ cannot be compared with any one among $\LmonoKromDia,\LmonoKromBox,$ $\LmonoKrom$, within the same propositional alphabet;
\item\label{cor313:point4} $\LmonoCoreBox\wlexp\LmonoHornBox$ and $\LmonoCoreDia\wlexp\LmonoHornDia$.

\end{compactenum}
\end{cor}

\begin{proof}
As we have seen in Theorem~\ref{th:horn_vs_bool}, the formula $p\vee q$, which belongs to all sub-Krom fragments of \LmonoBool, cannot be translated into \LmonoHorn, and, therefore, it cannot be translated into any sub-Horn and sub-core fragment either, even with an extended propositional alphabet (and, hence, even more so within the same alphabet). Since $\LmonoCoreBox$ is a fragment of $\LmonoKrom$ and of $\LmonoKromBox$, and since  $\LmonoCoreDia$ is a fragment of $\LmonoKrom$ and of $\LmonoKromDia$, Theorem~\ref{th:horn_vs_bool} proves that $\LmonoCoreBox\slexp\LmonoKrom$, $\LmonoKromBox$ and that $\LmonoCoreDia\slexp\LmonoKrom$, $\LmonoKromDia$. By Theorems~\ref{th:krombox_vs_krom} and \ref{th:kromdia_vs_boxdia} we have that the three fragments $\LmonoKrom$, $\LmonoKromBox$ and $\LmonoKromDia$ are model-extending equivalent, that is, there is a translation between any pair of them if we allow the use of additional propositional letters.
Since $\LmonoCoreBox \slexp \LmonoKrom$ and $\LmonoKrom \meequiv \LmonoKromDia$ we have that also $\LmonoCoreBox \slexp \LmonoKromDia$ holds, and by the very same reasoning we can prove that $\LmonoCoreDia \slexp \LmonoKromBox$.
This proves points \ref{cor313:point1} and \ref{cor313:point2}.
As for the other points, Theorem~\ref{th:krom_vs_bool} proves that the formula $(p\wedge q)\rightarrow r$, which belongs to all sub-Horn fragments of \LmonoBool, cannot be translated into \LmonoKrom, and, therefore, it cannot be translated into any sub-Krom fragment either, at least within the same propositional alphabet.
This observation, together with the observation made in points \ref{cor313:point1} and \ref{cor313:point2}, proves point \ref{cor313:point3}.
Finally, using again Theorem~\ref{th:krom_vs_bool}, and taking into account that $\LmonoCoreBox = \LmonoHornBox \cap \LmonoKromBox$ and that $\LmonoCoreDia = \LmonoHornDia \cap \LmonoKromDia$, point  \ref{cor313:point4} immediately follows.
\end{proof}

\begin{figure}[t!]

\begin{tikzpicture}[>=latex,line join=bevel,scale=0.65,thick]
\begin{scope}
  \pgfsetstrokecolor{black}
  \definecolor{strokecol}{rgb}{1.0,0.0,0.0};
  \pgfsetstrokecolor{strokecol}
  \draw (20.0bp,80.0bp) .. controls (20.0bp,80.0bp) and (213.0bp,80.0bp)  .. (213.0bp,80.0bp) .. controls (219.0bp,80.0bp) and (225.0bp,86.0bp)  .. (225.0bp,92.0bp) .. controls (225.0bp,92.0bp) and (225.0bp,192.0bp)  .. (225.0bp,192.0bp) .. controls (225.0bp,198.0bp) and (219.0bp,204.0bp)  .. (213.0bp,204.0bp) .. controls (213.0bp,204.0bp) and (20.0bp,204.0bp)  .. (20.0bp,204.0bp) .. controls (14.0bp,204.0bp) and (8.0bp,198.0bp)  .. (8.0bp,192.0bp) .. controls (8.0bp,192.0bp) and (8.0bp,92.0bp)  .. (8.0bp,92.0bp) .. controls (8.0bp,86.0bp) and (14.0bp,80.0bp)  .. (20.0bp,80.0bp);
\end{scope}
\begin{scope}
  \pgfsetstrokecolor{black}
  \definecolor{strokecol}{rgb}{1.0,0.0,0.0};
  \pgfsetstrokecolor{strokecol}
  \draw (20.0bp,80.0bp) .. controls (20.0bp,80.0bp) and (213.0bp,80.0bp)  .. (213.0bp,80.0bp) .. controls (219.0bp,80.0bp) and (225.0bp,86.0bp)  .. (225.0bp,92.0bp) .. controls (225.0bp,92.0bp) and (225.0bp,192.0bp)  .. (225.0bp,192.0bp) .. controls (225.0bp,198.0bp) and (219.0bp,204.0bp)  .. (213.0bp,204.0bp) .. controls (213.0bp,204.0bp) and (20.0bp,204.0bp)  .. (20.0bp,204.0bp) .. controls (14.0bp,204.0bp) and (8.0bp,198.0bp)  .. (8.0bp,192.0bp) .. controls (8.0bp,192.0bp) and (8.0bp,92.0bp)  .. (8.0bp,92.0bp) .. controls (8.0bp,86.0bp) and (14.0bp,80.0bp)  .. (20.0bp,80.0bp);
\end{scope}
  \node (Core) at (336.0bp,106.0bp) [draw,draw=none] {$\LCore$};
  \node (Bool) at (335.0bp,250.0bp) [draw,draw=none] {$\LBool$};
  \node (CoreBox) at (285.0bp,34.0bp) [draw,draw=none] {$\LCoreBox$};
  \node (Krom) at (177.0bp,178.0bp) [draw,draw=none] {$\LKrom$};
  \node (HornBox) at (500.0bp,106.0bp) [draw,draw=none] {$\LHornBox$};
  \node (Horn) at (494.0bp,178.0bp) [draw,draw=none] {$\LHorn$};
  \node (KromDia) at (172.0bp,106.0bp) [draw,draw=none] {$\LKromDia$};
  \node (HornDia) at (606.0bp,106.0bp) [draw,draw=none] {$\LHornDia$};
  \node (KromBox) at (63.0bp,106.0bp) [draw,draw=none] {$\LKromBox$};
  \node (CoreDia) at (389.0bp,34.0bp) [draw,draw=none] {$\LCoreDia$};
  \draw [->,dashed] (Krom) ..controls (175.21bp,151.98bp) and (174.55bp,142.71bp)  .. (KromDia);
  \draw [red,->,solid] (Horn) ..controls (496.14bp,151.98bp) and (496.94bp,142.71bp)  .. (HornBox);
  \draw [red,->,solid] (Horn) ..controls (535.91bp,150.81bp) and (553.92bp,139.55bp)  .. (HornDia);
  \draw [red,->,solid] (KromDia) ..controls (215.2bp,85.116bp) and (222.26bp,82.327bp)  .. (229.0bp,80.0bp) .. controls (273.42bp,64.667bp) and (288.24bp,65.91bp)  .. (CoreDia);
  \draw [->,dashed] (HornDia) ..controls (566.65bp,85.093bp) and (560.19bp,82.312bp)  .. (554.0bp,80.0bp) .. controls (517.07bp,66.216bp) and (474.06bp,54.626bp)  .. (CoreDia);
  \draw [->,dashed] (Bool) ..controls (279.09bp,224.23bp) and (246.9bp,209.97bp)  .. (Krom);
  \draw [red,->,solid] (KromBox) ..controls (104.67bp,85.136bp) and (111.49bp,82.341bp)  .. (118.0bp,80.0bp) .. controls (155.13bp,66.659bp) and (198.15bp,55.157bp)  .. (CoreBox);
\draw [->,red,solid] (Krom) ..controls (234.17bp,151.83bp) and (263.35bp,138.99bp)  .. (Core);
  \draw [->,dashed] (Horn) ..controls (438.57bp,152.44bp) and (408.8bp,139.25bp)  .. (Core);
  \draw [red,->,solid] (Core) ..controls (317.4bp,79.474bp) and (310.12bp,69.483bp)  .. (CoreBox);
  \draw [->,red,solid] (Bool) ..controls (391.5bp,224.12bp) and (425.12bp,209.32bp)  .. (Horn);
  \draw [->,dashed] (Krom) ..controls (134.21bp,150.72bp) and (115.67bp,139.34bp)  .. (KromBox);
  \draw [red,->,solid] (Core) ..controls (355.39bp,79.389bp) and (363.05bp,69.277bp)  .. (CoreDia);
  \draw [->,dashed] (HornBox) ..controls (456.57bp,85.077bp) and (449.63bp,82.322bp)  .. (443.0bp,80.0bp) .. controls (400.18bp,64.998bp) and (386.05bp,65.476bp)  .. (CoreBox);
%
%
	\node (Equiv) at (60.0bp,185.0bp) [red] {$\meequiv$};
	\draw [->,dashed] (10.0bp,250.0bp) -- (60.0bp,250.0bp) node[right] {\footnotesize{``weakly less expressive"}};
	\draw [red,->,solid] (10.0bp,230.0bp) -- (60.0bp,230.0bp) node[right, black] {\footnotesize{``strongly less expressive"}};
\end{tikzpicture}


\medskip

\begin{tikzpicture}[>=latex,line join=bevel,scale=0.65,thick]
\begin{scope}
  \pgfsetstrokecolor{black}
  \definecolor{strokecol}{rgb}{1.0,0.0,0.0};
  \pgfsetstrokecolor{strokecol}
  \draw (20.0bp,80.0bp) .. controls (20.0bp,80.0bp) and (213.0bp,80.0bp)  .. (213.0bp,80.0bp) .. controls (219.0bp,80.0bp) and (225.0bp,86.0bp)  .. (225.0bp,92.0bp) .. controls (225.0bp,92.0bp) and (225.0bp,192.0bp)  .. (225.0bp,192.0bp) .. controls (225.0bp,198.0bp) and (219.0bp,204.0bp)  .. (213.0bp,204.0bp) .. controls (213.0bp,204.0bp) and (20.0bp,204.0bp)  .. (20.0bp,204.0bp) .. controls (14.0bp,204.0bp) and (8.0bp,198.0bp)  .. (8.0bp,192.0bp) .. controls (8.0bp,192.0bp) and (8.0bp,92.0bp)  .. (8.0bp,92.0bp) .. controls (8.0bp,86.0bp) and (14.0bp,80.0bp)  .. (20.0bp,80.0bp);
\end{scope}
\begin{scope}
  \pgfsetstrokecolor{black}
  \definecolor{strokecol}{rgb}{1.0,0.0,0.0};
  \pgfsetstrokecolor{strokecol}
  \draw (20.0bp,80.0bp) .. controls (20.0bp,80.0bp) and (213.0bp,80.0bp)  .. (213.0bp,80.0bp) .. controls (219.0bp,80.0bp) and (225.0bp,86.0bp)  .. (225.0bp,92.0bp) .. controls (225.0bp,92.0bp) and (225.0bp,192.0bp)  .. (225.0bp,192.0bp) .. controls (225.0bp,198.0bp) and (219.0bp,204.0bp)  .. (213.0bp,204.0bp) .. controls (213.0bp,204.0bp) and (20.0bp,204.0bp)  .. (20.0bp,204.0bp) .. controls (14.0bp,204.0bp) and (8.0bp,198.0bp)  .. (8.0bp,192.0bp) .. controls (8.0bp,192.0bp) and (8.0bp,92.0bp)  .. (8.0bp,92.0bp) .. controls (8.0bp,86.0bp) and (14.0bp,80.0bp)  .. (20.0bp,80.0bp);
\end{scope}
  \node (Core) at (336.0bp,106.0bp) [draw,draw=none] {$\LCore$};
  \node (Bool) at (335.0bp,250.0bp) [draw,draw=none] {$\LBool$};
  \node (CoreBox) at (285.0bp,34.0bp) [draw,draw=none] {$\LCoreBox$};
  \node (Krom) at (177.0bp,178.0bp) [draw,draw=none] {$\LKrom$};
  \node (HornBox) at (500.0bp,106.0bp) [draw,draw=none] {$\LHornBox$};
  \node (Horn) at (494.0bp,178.0bp) [draw,draw=none] {$\LHorn$};
  \node (KromDia) at (172.0bp,106.0bp) [draw,draw=none] {$\LKromDia$};
  \node (HornDia) at (606.0bp,106.0bp) [draw,draw=none] {$\LHornDia$};
  \node (KromBox) at (63.0bp,106.0bp) [draw,draw=none] {$\LKromBox$};
  \node (CoreDia) at (389.0bp,34.0bp) [draw,draw=none] {$\LCoreDia$};
  \draw [->,dashed] (Krom) ..controls (175.21bp,151.98bp) and (174.55bp,142.71bp)  .. (KromDia);
  \draw [red,->,solid] (Horn) ..controls (496.14bp,151.98bp) and (496.94bp,142.71bp)  .. (HornBox);
reflex  \draw [->,dashed] (Horn) ..controls (535.91bp,150.81bp) and (553.92bp,139.55bp)  .. (HornDia);
reflex  \draw [->,red, solid] (KromDia) ..controls (215.2bp,85.116bp) and (222.26bp,82.327bp)  .. (229.0bp,80.0bp) .. controls (273.42bp,64.667bp) and (288.24bp,65.91bp)  .. (CoreDia);

  \draw [->,dashed] (HornDia) ..controls (566.65bp,85.093bp) and (560.19bp,82.312bp)  .. (554.0bp,80.0bp) .. controls (517.07bp,66.216bp) and (474.06bp,54.626bp)  .. (CoreDia);
  \draw [->,dashed] (Bool) ..controls (279.09bp,224.23bp) and (246.9bp,209.97bp)  .. (Krom);
  \draw [red,->,solid] (KromBox) ..controls (104.67bp,85.136bp) and (111.49bp,82.341bp)  .. (118.0bp,80.0bp) .. controls (155.13bp,66.659bp) and (198.15bp,55.157bp)  .. (CoreBox);
  \draw [->,red,solid] (Krom) ..controls (234.17bp,151.83bp) and (263.35bp,138.99bp)  .. (Core);
  \draw [->,dashed] (Horn) ..controls (438.57bp,152.44bp) and (408.8bp,139.25bp)  .. (Core);
  \draw [red,->,solid] (Core) ..controls (317.4bp,79.474bp) and (310.12bp,69.483bp)  .. (CoreBox);
 \draw [red,->,solid] (Bool) ..controls (391.5bp,224.12bp) and (425.12bp,209.32bp)  .. (Horn);
  \draw [->,dashed] (Krom) ..controls (134.21bp,150.72bp) and (115.67bp,139.34bp)  .. (KromBox);
  \draw [->,dashed] (Core) ..controls (355.39bp,79.389bp) and (363.05bp,69.277bp)  .. (CoreDia);
  \draw [->,dashed] (HornBox) ..controls (456.57bp,85.077bp) and (449.63bp,82.322bp)  .. (443.0bp,80.0bp) .. controls (400.18bp,64.998bp) and (386.05bp,65.476bp)  .. (CoreBox);
%
%
	\node (Equiv) at (60.0bp,185.0bp) [red] {$\meequiv$};
\end{tikzpicture}
	\caption{Relative expressive power: $\Lmono\in\{\Kmono, \Fmono\}$ (top), and $\Lmono\in\{\Tmono, \SFmono\}$ (bottom).}
   \label{fig:exprpowerComplete}
\end{figure}

\medskip

A graphical account of the results of this section is shown in Fig.~\ref{fig:exprpowerComplete}. Among all results, most striking is the lost of expressive power of the languages of the type \LmonoHornBox\ and \LmonoCoreBox. The 
important expressivity change occurs
at two different moments, depending on the class of frames: for non-reflexive classes it occurs below \LmonoHorn, while for reflexive classes it occurs between \LmonoBool\ and \LmonoHorn. Either way,  \LmonoHornBox\ and \LmonoCoreBox\ seem the best candidates to be further studied in terms of the complexity of their satisfiability problem.

\section{Complexity}


In this section we study the complexity of the satisfiability problem for \LmonoHornBox\ and \LmonoCoreBox and we give a modular algorithm for the satisfiability-checking of \LmonoHornBox-formulas.
We begin by proving that if a formula of the Horn box fragment of $\Kmono$, $\Tmono$, $\Fmono$, or $\SFmono$ is satisfiable, then it is satisfiable in a very simple (pre-linear) model of bounded size (Theorem~\ref{th:prelinear}). Then we exploit this result to show that the satisfiability problem for such fragments is \Pt-complete (Theorem~\ref{th:pcompleteness}).
The complexity result does not give us a direct satisfiability-checking procedure. Thus, we define a modular algorithm that builds a model for the formula in deterministic polynomial time. Such a procedure is made of a common part (Algorithms \ref{alg:hornboxsat} and \ref{alg:shorten}) and a specific ``saturation procedure'' that depends on the considered language (Algorithms \ref{alg:ksaturate}, \ref{alg:fsaturate}, \ref{alg:tsaturate} and \ref{alg:sfsaturate}).
Correctness, completeness and complexity of the algorithms are proved in Theorem~\ref{th:complexity}.
Finally, we devote the last part of the section to the study of the core box fragments. We extend the well-known 2SAT algorithm to \KmonoCoreBox, thus proving \NLOG-completess of this fragment (Theorem \ref{th:khornbox_nlog}), leaving the exact complexity of the other core box fragments as an open problem.

\medskip

To prove the small-model theorem for $\Kmono$, $\Tmono$, $\Fmono$, or $\SFmono$, we need to introduce some preliminary definitions. We say that a relational structure $(W,R)$ is \emph{pre-linear} if it is either a simple path or a lasso-shaped path, namely, if:
\begin{inparaenum} \item there exists a unique node $w_0$ (called {\em root}) such that, for every $w\in W$, $w_0 R^* w$;
\item every $w \in W$ has at most one $R$-successor.
\end{inparaenum}
(Recall that, for a relation $R$, we denote by 
$R^*$ its reflexive and transitive closure.)
In a pre-linear structure we can enumerate the worlds in $W$ as $w_0, w_1, \ldots$, where $w_0$ is the root and for each $k \geq 0$ we have that $w_{k+1}$ is the unique $R$-successor of $w_k$. In the following, we use the term {\em pre-linear model} to denote a model built on a pre-linear relational structure. The following lemma, given for \KmonoHornBox, will be later generalized to the other cases.

\begin{lem}\label{lem:simpleprelinear}
Let $\varphi$ be a \KmonoHornBox-formula. Then, $\varphi$ is satisfiable if and only if it is satisfiable in a pre-linear model.
\end{lem}

\begin{proof}
Let $\varphi$ be a satisfiable \KmonoHornBox-formula, and let $M=(\mathcal F,V)$, where $\mathcal F=(W,R)$, be a model that satisfies it at $w_0$. We can assume that $\mathcal F$ is connected~\cite{blackburn2006handbook}. We now systematically build a pre-linear model $M'$, with set of worlds $W'$, as follows. The worlds in $W'$ are sets of worlds in $W$. Let us first define $W_0=\{w_0'\}$, $w_0'=\{w_0\}$, $R_0=\emptyset$, and $V_0(w_0')=V(w_0)$.
Now, given the generic pre-linear model $M_k=(\mathcal F_k,V_k)$, where $\mathcal F_k=(W_k,R_k)$, we define the pre-linear model $M_{k+1}$ as follows. Consider the `last' world $w_k'\subseteq W$, and define $w_{k+1}'=\{w\mid w\in W\mbox{ and }\exists w'\in w_k'\mbox{ s.t. }(w',w)\in R\}$, $W_{k+1}=W_k\cup\{w_{k+1}'\}$,  $R_{k+1}=R_{k}\cup\{(w_k',w_{k+1}')\}$, and $V_{k+1}(w_{k+1}')=\bigcap_{w\in w_{k+1}'}V(w)$. Clearly, $M'=\bigcup_k M_k$ is a pre-linear model.

To conclude the proof we have to show that $M',w_0'\mmodels\varphi$. Assume, by way of contradiction, that this is not the case, that is, $M',w_0'\not\mmodels\varphi$. Since $\varphi$ can be thought as a conjunction of clauses, this means that $M',w_0'\not\mmodels\varphi_i$ for some clause $\varphi_i=\boxx^s(\lambda_1\wedge\ldots\wedge\lambda_n\rightarrow\lambda)$. This means that $M',w_s'\not\mmodels\lambda_1\wedge\ldots\wedge\lambda_n\rightarrow\lambda$, where $w_s'$ is the world reachable from $w_0$ in exactly $s$ steps. This implies that $M',w_s'\mmodels\lambda_1\wedge\ldots\wedge\lambda_n$ and $M',w_s'\not\mmodels\lambda$. Consider a positive literal $\lambda_j$ in the body of the clause. We want to prove that $M,w\mmodels\lambda_j$ for each $w\in w_s'$. If $\lambda_j=p$ or $\lambda_j=\top$, the result is trivial, as the valuation of $w_s'$ is the intersection of the valuation (in $M$) of every world in $w_s'$. If $\lambda_j=\boxx^m p$, then consider the world $w_{s+m}'$ reachable from $w_s'$ in exactly $m$ steps: two cases arise. If such a $w_{s+m}'$ exists, then $M',w_{s+m}'\mmodels p$, and this means that (in $M$)  every world $\overline w$ in $w_{s+m}'$ is such that $M,\overline w\mmodels p$, that is, $M,w\mmodels\boxx^m p$ for every $w\in w_s'$. Otherwise, if $w_{s+m}'$ does not exist, then no $w\in w_s'$ has any $m$-successor in $M$, which, again, implies that $M,w\mmodels\boxx^m p$ for every $w\in w_s'$, as we wanted. Since $M,w_0\mmodels\varphi_i$ by hypothesis, this means that for every $w\in w_s'$ it is the case that $M,w\mmodels \lambda_1\wedge\ldots\wedge\lambda_n\rightarrow\lambda$. We have just proved that $M,w\mmodels \lambda_j$ for every antecedent $\lambda_j$ and every world $w\in w_s'$: this implies that $M,w\mmodels \lambda$ for each $w\in w_s'$. In turn, this implies $M',w_s'\mmodels\lambda$, which is a contradiction.
\end{proof}

\medskip

\noindent In order to generalize the above construction to the languages \THornBox, \FHornBox, and \SFHornBox, we need to  introduce the notions of reflexive, transitive, and reflexive and transitive pre-linear models. We call \emph{pre-linear reflexive model} any model obtained from a pre-linear K-model by replacing the accessibility relation $R$ with its reflexive closure $R^\circlearrowright$; \emph{pre-linear transitive models}, as well as \emph{pre-linear transitive and reflexive models}, are defined likewise by replacing $R$ with $\overrightarrow{R}$ and $R^*$, respectively.

\begin{lem}\label{lem:allprelinear}
Let $\Lmono\in\{\Kmono,\Tmono, \Fmono, \SFmono\}$, and let $\varphi$ be  a \LHornBox-formula. Then $\varphi$ is satisfiable if and only if is satisfiable in a pre-linear model of the corresponding type.
\end{lem}

\begin{proof}
Lemma~\ref{lem:simpleprelinear} proves the claim for $\mathbf L = \Kmono$. We can suitably adapt the construction to prove the other cases. We first build a pre-linear model starting from the model $M$ that exists by hypothesis, using precisely the same technique as in Lemma~\ref{lem:simpleprelinear}. Once we have obtained the pre-linear model $M'$, we define its correspondent pre-linear reflexive, transitive, or reflexive and transitive model $M''$ by replacing the relation $R$ with its reflexive closure $R^\circlearrowright$, if $\mathbf L=\Tmono$, with its transitive closure  $\overrightarrow{R}$, if $\mathbf{L} = \Fmono$, or with its reflexive and transitive closure $R^*$, if  $\mathbf{L} = \SFmono$. The fact that $M'',w_0'\mmodels\varphi$ is then a straightforward adaptation of the case for \KmonoHornBox.
\end{proof}

\medskip

\noindent Lemma~\ref{lem:allprelinear} proves that we can restrict our attention to pre-linear models only. The following theorem restrict further the set of `relevant models' of a formula by providing a \emph{bounded pre-linear model property}: a formula is satisfiable if and only if there exists a pre-linear model of size bounded either by the modal depth of the formula (for \Kmono and \Tmono), or by the length of the formula (for \Fmono and \SFmono).

\begin{thm}\label{th:prelinear}
Let $\Lmono\in\{\Kmono,\Tmono, \Fmono, \SFmono\}$, and let $\varphi$ be a \LHornBox-formula. Then:
\medskip
\begin{compactenum}
	\item if $\mathbf{L} = \Kmono$ (resp. $\mathbf{L} = \Tmono$), then $\varphi$ is satisfiable if and only if it is satisfiable in a pre-linear (resp. pre-linear reflexive) model $M$ such that $|M| \leq md(\varphi)+1$;
	\item if $\mathbf{L} = \Fmono$ (resp. $\mathbf{L} = \SFmono$), then $\varphi$ is satisfiable if and only if it is satisfiable in a pre-linear transitive (resp. transitive and reflexive) model $M$ such that $|M| \leq |\varphi|$.
\end{compactenum}
\end{thm}

\begin{proof}
It is well-known that any \Kmono- and \Tmono-formula $\varphi$ that is satisfiable on a world $w_0$ of a model $M=(\mathcal F, V)$ is also satisfiable by the submodel $M'=(\mathcal F', V')$ of $M$ containing only the worlds accessible from $w_0$ in at most $md(\varphi)$ steps~\cite{blackburn2006handbook}.
By Lemma~\ref{lem:allprelinear} we know that when $\varphi$ is a satisfiable \LmonoHornBox-formula, with $\mathbf L\in \{\Kmono, \Tmono\}$, then we can assume $M$ to be a pre-linear $\Lmono$-model. Hence $M'$ is a pre-linear model of size less or equal to $md(\varphi) + 1$.

As far as the second claim is concerned, let us consider the case of $\FHornBox$ first. From Lemma~\ref{lem:allprelinear} we know that there exists a pre-linear \Fmono-model $M=(\mathcal F,V)$, where $\mathcal F=(W,R)$, that satisfies it at $w_0$. Suppose that $|W| > |\varphi|$: we now show how to build a shorter pre-linear \Fmono-model of the desired size. For every world $w \in W$ we define $B(w)$ as the set of $\boxx\xi$ formulas in the closure of $\varphi$ that holds on $w$: $B(w) = \{\boxx\xi \in Cl(\varphi) \mid M, w \mmodels \boxx\xi\}$. By transitivity, we have that for every pair of worlds $(w, w') \in R$, if $M, w \mmodels \boxx\xi$ then $M, w' \mmodels \boxx\xi$, and thus the sets $B(w)$ and $B(w')$ respect the following \emph{inclusion property}: if $(w, w') \in R$ then $B(w) \subseteq B(w')$. Given a world $w \in W$, we define its \emph{$\boxx$-equivalence class} as $\eqclass{w} = \{w' \in W \mid B(w) = B(w')\}$. Observe that the inclusion property guarantees that, for every $w_{k'},w_{k''}\in W$, if $w_{k'},w_{k''}\in \eqclass{w}$ for some $w$ and $k'<k''$, then, for every $k$ such that $k' < k < k''$, we have that $w_k\in\eqclass{w}$. By the inclusion property, we have that for every $k' < k''$, $B(w_{k'}) \subseteq B(w_{k''})$. This implies that, if $\eqclass{w_{k'}} \neq \eqclass{w_{k''}}$ then $|B(w_{k'})| < |B(w_{k''})|$. Since, for every $w \in W$, $|B(w)| \leq |\varphi| - 1$, we can conclude that the number of equivalence classes in $W$ is bounded by $|\varphi|$. Let us define the following elements: \begin{inparaenum} \item a set of worlds is $W' = \{\eqclass{w}\mid w \in W\}$;
\item an accessibility relation ${R'}$ defined as $\{(\eqclass{w},\eqclass{w'}) \mid (w, w') \in R, \eqclass{w}\neq \eqclass{w'}\}$ (notice that being $\Lmono=\Fmono$, and thanks to the fact that the model under analysis is pre-linear, every element of $\eqclass{w}$ is $R$-related to every element of $\eqclass{w'}$);
\item a valuation $V'$ such that $V'(\eqclass{w}) = \bigcap_{w' \in \eqclass{w}} V(w')$.
\end{inparaenum} Since $W'$ is always finite, two cases may arise: either the `last' world $\eqclass{w}$ is a singleton, or it is not; in the latter case, we add the pair $(\eqclass{w},\eqclass{w})$ to $R'$. We now prove that the model $M' = (\mathcal{F}',V')$ where $\mathcal{F}' = (W', R')$ is, indeed, a model for $\varphi$. To this end, suppose by contradiction that this is not the case, and that $M', \eqclass{w_0} \not\mmodels \varphi$. This implies that there exists a clause $\varphi_i = \boxx^s(\lambda_1 \land \ldots \land \lambda_n \rightarrow \lambda)$ such that $M', \eqclass{w_0} \not\mmodels \varphi_i$, namely, that there exists a $\eqclass{w}$ such that $(\eqclass{w_0},\eqclass{w})\in (R')^s$, $M', \eqclass{w} \mmodels \lambda_1 \land \ldots \land \lambda_n$ and $M', \eqclass{w} \not\mmodels \lambda$.
Let $\lambda_j$ be an arbitrary literal in the body of the clause: we first show that, by construction, we have that $M, w' \mmodels \lambda_j$ for every $w' \in \eqclass{w}$. If $\lambda_j = p$ or $\lambda_j = \top$ the proof is trivial, since $V'(\eqclass{w}) \subseteq V(w')$. If $\lambda_j = \boxx^m p$ then we have that $p \in V'(\eqclass{w''})$ for all $\eqclass{w''}$ such that $(\eqclass{w},\eqclass{w''}) \in (R')^m$, or there is no $\eqclass{w''}$ such that $(\eqclass{w},\eqclass{w''}) \in (R')^m$. In the former case, by the definition of $M'$ we have that $p \in V(w''')$ for all $w'''\in\eqclass{w''}$, and thus that $M, w' \mmodels \boxx^m p$. In the latter case, $\eqclass{w}$ is the `last' world of $M'$, and there is no self-loop, which implies that, by construction, $|\eqclass{w}|=1$; consider, now, $w$ in $M$: it has no successor, and, therefore, $M,w'\mmodels\boxx^m p$ for every (in fact, the only one) $w'\in \eqclass{w}$. In all cases, $M, w' \mmodels \lambda_j$. Since this holds for all $\lambda_j$ in the body of the clause, and since $M$ is a model for $\varphi$, we have that $M, w' \mmodels \lambda$. To prove that $M', \eqclass{w} \mmodels \lambda$, three cases arise: either $\lambda = p$, or $\lambda = \top$, or $\lambda = \boxx^m p$. If $\lambda = p$, or $\lambda = \top$ the proof follows by the definition of $V'$. If $\lambda = \boxx^m p$ then for every $w''$ such that $(w',w'') \in (R)^m$ we have that $p \in V(w'')$. But this implies that $p \in V'(\eqclass{w''})$ for every $\eqclass{w''}$ such that $(\eqclass{w},\eqclass{w''}) \in (R')^m$, which, in turn, means that $M',\eqclass{w} \mmodels \boxx^m p$. In all cases a contradiction is found.

Finally, to prove the claim in the case of $\SFHornBox$ we can use the very same argument, where the accessibility relation of the `shorter model' $M'$  is defined as $R' = \{(\eqclass{w},\eqclass{w'}) \mid (w, w') \in R\}$. Since we start from a reflexive and transitive pre-linear model, we have that $R'$ is also reflexive and transitive, and the claim follows.
\end{proof}

When paired with the results in~\cite{Nguyen2000, Nguyen2016}, Theorem~\ref{th:prelinear} allows us to prove \Pt-completeness of \KmonoHornBox, \TmonoHornBox, \FmonoHornBox, and \SFmonoHornBox, as proved by the following theorem\footnote{Thanks to the anonymous referee for pointing this out.}.

\begin{thm}\label{th:pcompleteness}
For every $\Lmono \in \{\Kmono, \Tmono, \Fmono, \SFmono\}$, the satisfiability problem for $\LHornBox$ is \Pt-complete.
\end{thm}

\begin{proof}
It is known that any Horn formula can be rewritten into a positive logic program (positive clauses) and a set of queries (negative clauses), and thus satisfiability can be reduced to checking the queries w.r.t.\ to the positive logic program.
\cite{Nguyen2000}~gives an algorithm for building the minimal model of a positive logic program of \TmonoHorn and \SFmonoHorn. When formulas are in the Horn-box fragment, by Theorem~\ref{th:prelinear} we have that such a minimal model is pre-linear and of size bounded by either the length or the modal depth of the formula. Thus, checking the queries can be done in polynomial deterministic time. Proving that the satisfiability problem of \KmonoHornBox and \FmonoHornBox is in \Pt\ requires a more involved reduction. The results of~\cite{Nguyen2000} do not apply, since they cover only serial and almost-serial modal logics. The work in~\cite{Nguyen2016} gives an
algorithm for checking the satisfiability of a description logic called Horn-CPDL$_{reg}$.
Such a logic subsumes \KmonoHorn and also \FmonoHorn, since transitivity is expressible in the language. When restricted to the Horn-box fragment, it is possible to show that the pseudo-model built by the algorithm respects Theorem~\ref{th:prelinear}, and thus that the complexity is in \Pt. Finally, for every $\Lmono \in \{\Kmono, \Tmono, \Fmono, \SFmono\}$, \Pt-hardness follows from the fact that propositional Horn is already \Pt-complete.
\end{proof}

The above theorem establishes the complexity class of the satisfiability problem for all the considered fragments by reduction to a more expressive language. This is sufficient to prove \Pt-completeness, but does not give a direct satisfiability-checking procedure.
Hence, in the following we give a modular algorithm that builds a model for the formula without the need of translating it into a different language, and that it is simple and easy to implement.
Algorithms~\ref{alg:hornboxsat}--\ref{alg:sfsaturate} defines such a modular procedure, and check the satisfiability of a set of clauses of \LHornBox, for $\mathbf L\in\{\Kmono, \Tmono, \Fmono, \SFmono\}$, in polynomial time. We first prove its correctness and completeness in the simpler case of \KmonoHornBox, and then we specialise the argument for \FmonoHornBox. We skip the formal proofs for \TmonoHornBox and \SFmonoHornBox since, as discussed above, they follows directly from~\cite{Nguyen2000}.

Our decision procedure iteratively builds a structure $(W, H, L)$ that represents a candidate model for the formula, where $W$ is the set of world of the model, $H$ labels each world in $W$ with the set of formulas `to be further analyzed' and  $L$ labels each world $w$ with the set of formulas in $Cl(\varphi)$ that holds on $w$. We first prove that if \Call{HornBoxSat}{} terminates with success, then the final structure indeed represents a model for $\varphi$.

\begin{algorithm}[t]
\caption{Main algorithm.}\label{alg:hornboxsat}
\begin{algorithmic}[1]
 \Function{HornBoxSat}{$\Lmono,\varphi$}
 \State{$W \gets \{w_0\}$}
 \State{$H(w_0) \gets \{\varphi_1, \ldots, \varphi_l\}$}
 \State{$L(w_0) \gets \{\top\}$}
    \If{$\mathbf L=\Kmono$ {\bf or} $\mathbf L=\Tmono$}
        \State{$N=md(\varphi)$}
    \ElsIf{$\mathbf L=\Fmono$ {\bf or} $\mathbf L=\SFmono$}
        \State{$N=|\varphi|$}
  \EndIf
 \For{$d \gets 0, \ldots, N$}\label{main:for}
  \If {\Call{\Lmono-Saturate}{$W, H, L$}} \Comment{{\em add a world}}\label{main:saturate}
   \State{$W \gets W \cup \{w_{d+1}\}$}
   \State{$H(w_{d+1}) \gets \emptyset$}
   \State{$L(w_{d+1}) \gets \{\top\}$}
  \Else \Comment{{\em inconsistency: try a smaller model}}
   \State \Return{\Call{Shorten}{$\Lmono, W, H, L, d$}}
  \EndIf
 \EndFor\label{main:endfor}
  \State \Return{True}
\EndFunction
\end{algorithmic}
\end{algorithm}

\begin{algorithm}[t!]
\caption{Saturation procedure for \Kmono.}\label{alg:ksaturate}
\begin{algorithmic}[1]
	\Function{\Kmono-Saturate}{$W, H, L$}
		\While{something changes}
			\State{\textbf{let} $w_k \in W$, $\psi \in H(w_k)$} \label{ksat:let}
			\If{$w_{k-1} \in W$ \textbf{and} $\boxx \psi \in Cl(\varphi)$}\label{ksat:ifprec}
					\State{$H(w_{k-1}) \gets H(w_{k-1}) \cup \{\boxx \psi\}$}
			\EndIf \label{ksat:endifprec}
			\If{$\psi = p$} \label{ksat:ifprop}
				\State{$H(w_k) \gets H(w_k) \setminus \{p\}$}
				\State{$L(w_k) \gets L(w_k) \cup \{p\}$} \label{ksat:addlprop} \label{ksat:endifprop}
			\ElsIf{$\psi = \boxx \xi$}\label{ksat:ifbox}
				\If{$w_{k+1} \in W$}\label{ksat:ifsuccbox}
					\State{$H(w_k) \gets H(w_k) \setminus \{\boxx\xi\}$}
                    \State{$L(w_k) \gets L(w_k) \cup \{\boxx\xi\}$}
					\State{$H(w_{k+1}) \gets H(w_{k+1}) \cup \{\xi\}$}
				\EndIf \label{ksat:endifsuccbox}\label{ksat:endifbox}
			\ElsIf{$\psi = \lambda_1 \land \ldots \land \lambda_n \rightarrow \lambda$, $\lambda \neq \bot$}
				\label{ksat:ifposcl}
				\If{$\{\lambda_1, \ldots, \lambda_n\} \subseteq H(w_k) \cup L(w_k)$}
					\State{$H(w_k) \gets \left(H(w_k) \cup \{\lambda\}\right) \setminus \{\psi\}$}
					\State{$L(w_k) \gets L(w_k) \cup \{\psi\}$}
				\EndIf \label{ksat:endifposcl}
			\ElsIf{$\psi = \lambda_1 \land \ldots \land \lambda_n \rightarrow \bot$} \label{ksat:ifnegcl}
				\If{$\{\lambda_1, \ldots, \lambda_n\} \subseteq H(w_k) \cup L(w_k)$}
					\State \Return{False}
				\EndIf				
			\EndIf\label{ksat:endifnegcl}
		\EndWhile
		\State \Return{True}
	\EndFunction
\end{algorithmic}
\end{algorithm}

\begin{algorithm}[t]
\caption{Shortening procedure.}\label{alg:shorten}
\begin{algorithmic}[1]
	\Function{Shorten}{$\Lmono, W, H, L, d$}
		\State{$B \gets \{\boxx \lambda \mid \boxx \lambda\in Cl(\varphi)\}$}
		\For{$k \gets d, d-1, \ldots, 1$}
			\State{$W \gets W \setminus \{w_{k}\}$}
			\State{$H \gets H|_W$, $L \gets L|_W$}
			\State{$H(w_{k-1}) \gets H(w_{k-1}) \cup B$}\label{sh:addboxes}
			\If{\Call{\Lmono-Saturate}{$W, H, L$}}\label{sh:saturate}
				\State \Return{True}
			\EndIf
		\EndFor
		\State\Return{False}
	\EndFunction
\end{algorithmic}
\end{algorithm}

\begin{lem}\label{lem:k-correctness}
Let $\varphi$ be a \KmonoHornBox-formula. If \Call{HornBoxSat}{$\Kmono,\varphi$} returns True, then $\varphi$ is satisfiable.
\end{lem}

\begin{proof}
Assume that \Call{HornBoxSat}{$\Kmono,\varphi$} returned True, and consider the triple $(W, H, L)$ as it has been built at the end of the execution. As before, we assume w.l.o.g.\ that $\varphi = \varphi_1 \land \ldots \land \varphi_l$ where each $\varphi_i$ is a clause.
We define a model $M$ based on the frame $(W,R)$, where $R = \{(w_k, w_{k+1}) \mid 0 \leq k < |W| - 1\}$, and such that, for every $w\in W$, $V(w)=L(w)|_{\mathcal P}$. We want to prove that $M,w_0\mmodels\varphi$, where $w_0$ is the root of the pre-linear model $M$. Suppose, by way of contradiction, that this is not the case, that is, suppose that $M,w_0\not\mmodels\varphi_i$ for some $\varphi_i$. Assume that $\varphi_i=\boxx^s(\lambda_1\wedge\ldots\wedge\lambda_n\rightarrow\lambda)$, and let $w_0,w_1,\ldots$ be the enumeration of the of worlds in $W$. Clearly, $w_s\in W$, as in the opposite case, the clause would be trivially satisfied. Hence, we have that $M, w_s \mmodels \lambda_1\wedge\ldots\wedge\lambda_n$ but $M, w_s \not\mmodels \lambda$. Since $w_s\in W$, the {\bf for} cycle of the main procedure has run enough times to create $w_s$, and since \Call{\Kmono-Saturate}{} has been executed after that $w_s$ has been created, $(\lambda_1\wedge\ldots\wedge\lambda_n\rightarrow\lambda)$ has been inserted in $H(w_s)$ (lines \ref{ksat:ifbox}-\ref{ksat:endifsuccbox}). We first show that, for every positive literal $\lambda_j$ in the head of the clause, we have $\lambda_j \in H(w_s) \cup L(w_s)$. If $\lambda_j = p$, then, since $M, w_s \mmodels p$ we have that $p \in L(w_s)$ by the definition of $M$. If $\lambda_j = \boxx^m p$, since we know that $M,w_s\mmodels\lambda_j$, it must be that either $M,w_{s+m}\mmodels p$, that is, $p\in L(w_{s+m})$, or  $s+m>|W|-1$. Suppose, first, that $s+m>|W|-1$. This means that the {\bf for} cycle of the main procedure has been interrupted before reaching the value $md(\varphi)$ (which is the limit in the case of $\mathbf L=\Kmono$), and that \Call{Shorten}{} has been executed, returning True. Observe that \Call{Shorten}{} may have eliminated several worlds, but immediately before returning True, it has eliminated the world $w_{|W|}$. At precisely this moment, $\boxx p$, $\boxx\boxx p$,\ldots,$\boxx^m p$ have been inserted into $H(w_{|W|-1})$ (line~\ref{sh:addboxes}). Since $s\le |W|-1<s+m$, we have that $(|W|-1)-s\le m$, which, in turn, implies that the last execution of \Call{\Kmono-Saturate}{} must have inserted $\boxx^m p\in H(w_s)$ (lines~\ref{ksat:ifprec}-\ref{ksat:endifprec}). When $w_{s+m} \in W$, a similar argument proves, again, that $\boxx^m p\in H(w_s)$. Either way, $\lambda_j\in L(w_s)$.
Since every positive literal in the body of the clause $(\lambda_1\wedge\ldots\wedge\lambda_n\rightarrow\lambda)$ belongs to $H(w_s)$, we have that some execution of \Call{\Kmono-Saturate}{} must have put $\lambda\in H(w_s)$ (lines~\ref{ksat:ifposcl}-\ref{ksat:endifposcl}).
If $\lambda=p$, within the same execution of \Call{\Kmono-Saturate}{} we would have that $p\in L(w_s)$ (lines~\ref{ksat:ifprop}-\ref{ksat:addlprop}), and thus that $M, w_s \mmodels p$, which contradicts the hypothesis that $M$ does not satisfy the formula. Now, if $\lambda=\boxx^m p$, then, since $|W|\ge s+m+1$ (otherwise $M,w_s\mmodels\lambda$ trivially, which is in contradiction with our hypothesis), the execution of \Call{\Kmono-Saturate}{}  that takes place at the moment of creation of the point $w_{s+m}$ puts $p\in L(w_{s+m})$, which, again, contradicts our hypothesis. In all cases we have a contradiction and thus we have proved that $M, w_0 \mmodels \varphi$.
\end{proof}

\medskip

\noindent As far as the completeness is concerned, it is convenient to introduce the following definitions. Let $M$ be a pre-linear model built on a frame $(W', R')$. We can assume w.l.o.g.\ that the frame a single path with no loops: if $M$ is a looping model, we can unfold it to obtain a model based on a single (infinite) path.
We say that a structure $(W, H, L)$ \emph{is compatible} with $M$ if, for every $i < \min(|W|,|W'|)$ we have that $\psi \in H(w_i) \cup L(w_i)$ implies that $M, w_i' \mmodels \psi$. Notice that the notion of compatibility compares structures and models of different size: in particular, it may be the case that the structure contains more worlds than $M$, and in this case, nothing is required on the `excess worlds' that do not belong to the model.
To prove completeness of the algorithm, we need to prove some relevant properties of compatible structures and of the \Call{\Kmono-Saturate}{} function.

\begin{lem}\label{lem:k-saturate}
Let $M$ be a pre-linear model, and let $(W, H, L)$ be a structure compatible with $M$. Then the structure obtained after the execution of \Call{\Kmono-Saturate}{$W, H, L$} is compatible with $M$.
\end{lem}

\begin{proof}
We prove the claim by showing that after every iteration of the \textbf{while} loop the structure $(W, H, L)$ remains compatible with $M$. From now on, we identify the worlds of the structure $(W, H, L)$ as $w_0, w_1, \ldots$, and the worlds in $M$ as $w_0', w_1', \ldots$.
Let $w_k \in W$ and $\psi \in H(w_k)$ be respectively the world and the formula selected by the \textbf{let} statement of line~\ref{ksat:let}. If the corresponding $w_k'$ is not a world of $M$, then the property trivially holds. Otherwise, we proceed by following the pseudocode of \Call{\Kmono-Saturate}{}.
After selecting the world $w_k$ and the formula $\psi$, the iteration of the \textbf{while} loop proceeds with lines~\ref{ksat:ifprec}--\ref{ksat:endifprec} that inserts $\boxx \psi\in H(w_{k-1})$ if the predecessor $w_{k-1}$ of $w_k$ exists and $\boxx\psi \in Cl(\varphi)$. By inductive hypothesis, since $\psi$ was already in $H(w)\cup L(w)$, it was already true that $M,w_k'\mmodels\psi$.
Since $M$ is a pre-linear model, we also have that $M,w_{k-1}' \mmodels \boxx \psi$ (if $w_{k}' \neq w_0'$), as required by the fact that now $\boxx\psi \in H(w_{k-1})\cup L(w_{k-1})$.
The other steps depends on the structure of $\psi$.
If $\psi=p$, then lines~\ref{ksat:ifprop}--\ref{ksat:endifprop} move $\psi$ from $H(w_k)$ to $L(w_k)$.
By inductive hypothesis, since $p$ was already in $H(w)\cup L(w)$, it was already true that $M,w_k'\mmodels p$.
If $\psi=\boxx\xi$, then two cases arise. If there is no successor of $w_k$, then nothing changes and the property is trivially respected. If there exists a successor $w_{k+1}$ of $w_k$, then lines~\ref{ksat:ifsuccbox}--\ref{ksat:endifsuccbox} move $\boxx\xi$ from $H(w_k)$ to $L(w_k)$, and insert $\xi\in H(w_{k+1})$. By inductive hypothesis, $M,w_k'\mmodels\boxx\xi$, and, if the successor $w_{k+1}'$ exists in $M$, $M,w_{k+1}'\mmodels\xi$, as required (when the successor exists in the structure but not in $M$ nothing is required). If $\psi=\lambda_1\wedge\ldots\wedge\lambda_n\rightarrow \lambda$, then two cases arise. If $\lambda_j\notin H(w_k)\cup L(w_k)$ for some $\lambda_j$, then the structure remains unchanged, and the property trivially follows. If, on the other hand, each $\lambda_j\in H(w_k)\cup L(w_k)$, lines~\ref{ksat:ifposcl}--\ref{ksat:endifposcl} move $\lambda_1\wedge\ldots\wedge\lambda_n\rightarrow \lambda$ from $H(w_k)$ to $L(w_k)$ and insert $\lambda\in H(w_k)$. By inductive hypothesis, $M,w_k'\mmodels (\lambda_1\wedge\ldots\wedge\lambda_n\rightarrow \lambda)$ and $M,w_k'\mmodels \lambda_j$ for each antecedent
$\lambda_j$, which means that $M,w_k'\mmodels \lambda$, as we wanted.
Finally, if $\psi=\lambda_1\wedge\ldots\wedge\lambda_n\rightarrow \bot$, then we know that it is not the case that each antecedent $\lambda_j$ is in $H(w_k)\cup L(w_k)$ (otherwise, by inductive hypothesis, we would have that $M,w_k'\mmodels \bot$, which is a contradiction). So, since there exists at least one $\lambda_j\notin H(w_k)\cup L(w_k)$, the labelling does not change, and the property trivially follows.
\end{proof}

\medskip

\begin{lem}\label{lem:k-completeness}
Let $\varphi$ be a \KmonoHornBox-formula. If $\varphi$ is satisfiable, then \Call{HornBoxSat}{$\Kmono,\varphi$} returns True.
\end{lem}

\begin{proof}
Assume that $\varphi=\varphi_1\land\ldots\land\varphi_l$ is satisfiable: we want to prove that \Call{HornBoxSat}{$\Kmono,\varphi$} returns True. Thanks to Theorem~\ref{th:prelinear}, we can then ensure that $\varphi$ has a pre-linear model of size $md(\varphi)+1$ at most. Among all such pre-linear models that satisfy $\varphi$, consider the biggest one, and let us call it $M$ (with root $w_0$). Without losing any generality, we can assume that $M$ has no loops\footnote{If $M$ has a loop, we can unfold it and then take the submodel $M'$ containing only the worlds accessible from $w_0$ in at most $md(\varphi)$ steps: we know from~\cite{blackburn2006handbook} that $M'$ is still a model of the formula.}, and we denote the worlds $w_0,w_1,\ldots$ in $M$ and in the structure $(W,H,L)$ which  is built during the execution of \Call{HornBoxSat}{$\Kmono,\varphi$} in the same way. We now prove that at every execution step the structure $(W,H,L)$ is compatible with $M$. Lines 2-4 of \Call{HornBoxSat}{} initialize $(W,H,L)$ to the one-world structure such that $H(w_0) = \{\varphi_1,\ldots,\varphi_l\}$ and $L(w_0) = \{\top\}$. Since $M,w_0\mmodels \varphi$, the initial structure is trivially compatible with $M$.
Consider, now, the generic iteration of the {\bf for} loop of lines~\ref{main:for}--\ref{main:endfor}. By Lemma~\ref{lem:k-saturate} we have that the call to \Call{\Kmono-Saturate}{$W, H, L$} (line~\ref{main:saturate}) builds a compatible structure. If it returns True, then a new point $w_{d+1}$ is added and labelled with $\top$: the resulting $(W, H, L)$ is trivially compatible.
Assume, now, that $|W|=md(\varphi)+1$: in this case, we have just proved that \Call{\Kmono-Saturate}{} never returned False, and that the {\bf for} loop of the main procedure comes to an end. This implies that \Call{HornBoxSat}{$\Kmono,\varphi$} returns True, and the claim is proved. On the other hand, assume that  $|W|<md(\varphi)+1$. Two cases may arise: either the {\bf for} loop of the main procedure comes to an end, and \Call{HornBoxSat}{$\Kmono,\varphi$} returns True, as required, or at some iteration \Call{\Kmono-Saturate}{} returns False.
Since the structure $(W, H, L)$ is always compatible, this second case can only happen for some value of $d > |W|$, and \Call{Shorten}{} is executed on a structure bigger than $M$. The procedure starts eliminating worlds from the last one. Let us focus on the elimination of the world $w_{|W|}$ (the last ``excess world'' of the structure): line~\ref{sh:addboxes} of \Call{Shorten}{} adds every positive literal of the type $\boxx^m \lambda\in Cl(\varphi)$ ($m\ge 1$) in $H(w_{|W|-1})$. Now, observe that $M,w_{|W|-1}\mmodels\boxx^m\lambda$ because $w_{|W|-1}$ is the last world; therefore, the structure $(W, H, L)$ is still compatible. As a consequence, the execution of \Call{Saturate}{} (line~\ref{sh:saturate}) must return True, \Call{Shorten}{} terminates and returns True and hence \Call{HornBoxSat}{$\Kmono, \varphi$} returns True as well.
\end{proof}

\medskip

\begin{algorithm}[t]
\caption{Saturation procedure for \Fmono.}\label{alg:fsaturate}
\begin{algorithmic}[1]
	\Function{\Fmono-Saturate}{$W, H, L$}
		\While{something changes}
			\State{\textbf{let} $w_k \in W$, $\psi \in H(w_k)$} \label{fsat:let}
			\color{blue}
			\If{$w_k = w_N$  \textbf{and} $\boxx\psi \in Cl(\varphi)$}\label{fsat:ifloop}
					\State{$H(w_k) \gets H(w_k) \cup \{\boxx\psi\}$}
				\EndIf \label{fsat:endifloop}\color{black}		
			\If{$w_{k-1} \in W$ \textbf{and} $\boxx\psi \in Cl(\varphi)$}\label{fsat:ifprec}
				\color{blue}\If{$\boxx\psi \in H(w_{k}) \cup L(w_{k})$}\label{fsat:ifprectrans}
					\State{$H(w_{k-1}) \gets H(w_{k-1}) \cup \{\boxx\psi\}$}
				\EndIf			\color{black}\label{fsat:endifprectrans}
			\EndIf\label{fsat:endifprec}
			\If{$\psi = p$} \label{fsat:ifprop}
				\State{$H(w_k) \gets H(w_k) \setminus \{p\}$}
				\State{$L(w_k) \gets L(w_k) \cup \{p\}$} \label{fsat:addlprop}\label{fsat:endifprop}
			\ElsIf{$\psi = \boxx \xi$}\label{fsat:ifbox}
				\If{$w_{k+1} \in W$}\label{fsat:ifsuccbox}
					\State{$H(w_k) \gets H(w_k) \setminus \{\boxx\xi\}$}
                    \State{$L(w_k) \gets L(w_k) \cup \{\boxx\xi\}$}
					{\color{blue}\State{$H(w_{k+1}) \gets H(w_{k+1}) \cup \{\xi, \boxx\xi\}$}\label{fsat:boxtrans}}
				\EndIf \label{fsat:endifsuccbox}\label{fsat:endifbox}
			\ElsIf{$\psi = \lambda_1 \land \ldots \land \lambda_n \rightarrow \lambda$, $\lambda \neq \bot$}
				\label{fsat:ifposcl}
				\If{$\{\lambda_1, \ldots, \lambda_n\} \subseteq H(w_k) \cup L(w_k)$}
					\State{$H(w_k) \gets \left(H(w_k) \cup \{\lambda\}\right) \setminus \{\psi\}$}
					\State{$L(w_k) \gets L(w_k) \cup \{\psi\}$}
				\EndIf \label{fsat:endifposcl}
			\ElsIf{$\psi = \lambda_1 \land \ldots \land \lambda_n \rightarrow \bot$} \label{fsat:ifnegcl}
				\If{$\{\lambda_1, \ldots, \lambda_n\} \subseteq H(w_k) \cup L(w_k)$}
					\State \Return{False}
				\EndIf				
			\EndIf\label{fsat:endifnegcl}
		\EndWhile
		\State \Return{True}
	\EndFunction
\end{algorithmic}
\end{algorithm}

As we have already mentioned, \Call{HornBoxSat}{} can test the satisfiability of sets of clauses also in the cases of the axiomatic extensions \Tmono, \Fmono, and \SFmono. The main procedure remains the same, while the number of iterations of the \textbf{for} loop, as well as the saturation procedure, change. Algorithms \ref{alg:fsaturate}--\ref{alg:sfsaturate} show the pseudocode of the specialized \Call{\Lmono-Saturate}{} procedures. They differ from \Call{\Kmono-Saturate}{} and from each other by the lines of code that are highlighted, that implement rules for the \boxx-formulas that reflect the different properties of the accessibility relation (reflexive, transitive, reflexive and transitive). Consider for instance the procedure \Call{\Fmono-Saturate}{} described in Algorithm~\ref{alg:fsaturate}. Lines~\ref{fsat:ifloop}--\ref{fsat:endifloop} are added to deal with the self-loop in the last world $w_N$ of the model, that allows the procedure to represent an infinite model. Lines~\ref{fsat:ifprectrans}--\ref{fsat:endifprectrans} propagate \boxx-formulas to the predecessor of the current world respecting transitivity ($\boxx\psi$ holds on $w_{k-1}$ if both $\psi$ and $\boxx\psi$ hold on $w_k$), while line~\ref{fsat:boxtrans} propagates \boxx-formulas forward respecting transitivity (if $\boxx\psi$ holds on the current world, then both $\psi$ and $\boxx\psi$ hold on the successor). In the case of \Call{\Tmono-Saturate}{}, backward and forward propagation of \boxx-formulas are modified to respect reflexivity ($\boxx\psi$ holds on a world $w_k$, iff $\psi$ holds on both $w_k$ and its successor $w_{k+1}$), while \Call{\SFmono-Saturate}{} is obtained by merging the two saturation procedures for \Fmono and \Tmono.

Correctness and completeness of \Call{HornBoxSat}{} for the three axiomatic extensions can be proved by adapting the proofs for the basic language \Kmono. Let us start with correctness.

\begin{algorithm}[t]
\caption{Saturation procedure for \Tmono.}\label{alg:tsaturate}
\begin{algorithmic}[1]
	\Function{\Tmono-Saturate}{$W, H, L$}
		\While{something changes}
			\State{\textbf{let} $w_k \in W$, $\psi \in H(w_k)$} \label{tsat:let}
			\If{$w_{k-1} \in W$ \textbf{and} $\boxx\psi \in Cl(\varphi)$}\label{tsat:ifprec}
					\color{blue}\If{$\psi \in H(w_{k-1}) \cup L(w_{k-1})$}
						\State{$H(w_{k-1}) \gets H(w_{k-1}) \cup \{\boxx\psi\}$}
					\EndIf
			\EndIf\label{tsat:endifprec}\color{black}
			\If{$\psi = p$} \label{tsat:ifprop}
				\State{$H(w_k) \gets H(w_k) \setminus \{p\}$}
				\State{$L(w_k) \gets L(w_k) \cup \{p\}$} \label{tsat:addlprop} \label{tsat:endifprop}
			\ElsIf{$\psi = \boxx \xi$}\label{tsat:ifbox}
				{\color{blue}\State{$H(w_{k}) \gets H(w_{k}) \cup \{\xi\}$}\label{tsat:boxrefl}}
				\If{$w_{k+1} \in W$}\label{tsat:ifsuccbox}
					\State{$H(w_k) \gets H(w_k) \setminus \{\boxx\xi\}$}
                    \State{$L(w_k) \gets L(w_k) \cup \{\boxx\xi\}$}
					\State{$H(w_{k+1}) \gets H(w_{k+1}) \cup \{\xi\}$}
				\EndIf \label{tsat:endifsuccbox}\label{tsat:endifbox}
			\ElsIf{$\psi = \lambda_1 \land \ldots \land \lambda_n \rightarrow \lambda$, $\lambda \neq \bot$}
				\label{tsat:ifposcl}
				\If{$\{\lambda_1, \ldots, \lambda_n\} \subseteq H(w_k) \cup L(w_k)$}
					\State{$H(w_k) \gets \left(H(w_k) \cup \{\lambda\}\right) \setminus \{\psi\}$}
					\State{$L(w_k) \gets L(w_k) \cup \{\psi\}$}
				\EndIf \label{tsat:endifposcl}
			\ElsIf{$\psi = \lambda_1 \land \ldots \land \lambda_n \rightarrow \bot$} \label{tsat:ifnegcl}
				\If{$\{\lambda_1, \ldots, \lambda_n\} \subseteq H(w_k) \cup L(w_k)$}
					\State \Return{False}
				\EndIf				
			\EndIf\label{tsat:endifnegcl}
		\EndWhile
		\State \Return{True}
	\EndFunction
\end{algorithmic}
\end{algorithm}

\begin{algorithm}[t]
\caption{Saturation procedure for \SFmono.}\label{alg:sfsaturate}
\begin{algorithmic}[1]
	\Function{\SFmono-Saturate}{$W, H, L$}
		\While{something changes}
			\State{\textbf{let} $w_k \in W$, $\psi \in H(w_k)$} \label{sfsat:let}
			\If{$w_{k-1} \in W$ \textbf{and} $\boxx\psi \in Cl(\varphi)$}\label{sfsat:ifprec}
				\color{blue}\If{$\boxx\psi \in H(w_{k}) \cup L(w_{k})$}
					\If{$\psi \in H(w_{k-1}) \cup L(w_{k-1})$}
						\State{$H(w_{i-1}) \gets H(w_{i-1}) \cup \{\boxx\psi\}$}
					\EndIf
				\EndIf\color{black}
			\EndIf\label{sfsat:endifprec}
			\If{$\psi = p$} \label{sfsat:ifprop}
				\State{$H(w_k) \gets H(w_k) \setminus \{p\}$}
				\State{$L(w_k) \gets L(w_k) \cup \{p\}$} \label{sfsat:addlprop}
			\ElsIf{$\psi = \boxx \xi$}\label{sfsat:ifbox}
				{\color{blue}\State{$H(w_{k}) \gets H(w_{k}) \cup \{\xi\}$}}
				\If{$w_{k+1} \in W$}
					\State{$H(w_k) \gets H(w_k) \setminus \{\boxx\xi\}$}
					\State{$L(w_k) \gets L(w_k) \cup \{\boxx\xi\}$}
          {\color{blue}\State{$H(w_{k+1}) \gets H(w_{k+1}) \cup \{\xi, \boxx\xi\}$}}
				\EndIf \label{sfsat:endifsuccbox}\label{sfsat:endifbox}
			\ElsIf{$\psi = \lambda_1 \land \ldots \land \lambda_n \rightarrow \lambda$, $\lambda \neq \bot$}
				\label{sfsat:ifposcl}
				\If{$\{\lambda_1, \ldots, \lambda_n\} \subseteq H(w_k) \cup L(w_k)$}
					\State{$H(w_k) \gets \left(H(w_k) \cup \{\lambda\}\right) \setminus \{\psi\}$}
					\State{$L(w_k) \gets L(w_k) \cup \{\psi\}$}
				\EndIf \label{sfsat:endifposcl}
			\ElsIf{$\psi = \lambda_1 \land \ldots \land \lambda_n \rightarrow \bot$} \label{sfsat:ifnegcl}
				\If{$\{\lambda_1, \ldots, \lambda_n\} \subseteq H(w_k) \cup L(w_k)$}
					\State \Return{False}
				\EndIf				
			\EndIf\label{sfsat:endifnegcl}
		\EndWhile
		\State \Return{True}
	\EndFunction
\end{algorithmic}
\end{algorithm}

\begin{lem}\label{lem:l-correctness}
Let $\Lmono \in \{\Tmono, \Fmono, \SFmono\}$ and let $\varphi$ be a \LHornBox-formula. If \Call{HornBoxSat}{$\Lmono,\varphi$} returns True, then $\varphi$ is \LHornBox-satisfiable.
\end{lem}

\begin{proof}
To prove the claim we assume that $\varphi=\varphi_1\land\ldots\land\varphi_l$ and we consider the three cases of $\Lmono = \Tmono$, $\Lmono = \Fmono$, and $\Lmono = \SFmono$.

Let us prove, first, the case $\Lmono = \Tmono$.  Assume that \Call{HornBoxSat}{$\Tmono,\varphi$} returned True, and consider the structure $(W, H, L)$ as it has been built at the end of the execution.
We define  a model $M$ based on the frame $(W,R^\circlearrowright)$, where $R^\circlearrowright$ is the reflexive closure of the relation $R = \{(w_k, w_{k+1}),  \mid 0 \leq k < |W| - 1\}$, and such that, for every $w\in W$, $V(w)=L(w)|_{\mathcal P}$. We want to prove that $M,w_0\mmodels\varphi$, where $w_0$ is the root of the pre-linear \Tmono-model $M$. Suppose, by way of contradiction, that this is not the case, that is, suppose that $M,w_0\not\mmodels\varphi_i$ for some $\varphi_i$.
Assume that $\varphi_i=\boxx^s(\lambda_1\wedge\ldots\wedge\lambda_n\rightarrow\lambda)$, and let $w_0,w_1,\ldots$ be the enumeration of the of worlds in $W$.
By the reflexivity of $R^\circlearrowright$ we have that there exists $w_t\in W$, with $t \leq s$, such that $M, w_t \mmodels \lambda_1\wedge\ldots\wedge\lambda_n$ but $M, w_t \not\mmodels \lambda$.
Since $w_t\in W$, the {\bf for} cycle of the main procedure has run enough times to create $w_t$, and since \Call{\Tmono-Saturate}{} has been executed when  $w_t$ has been created, $(\lambda_1\wedge\ldots\wedge\lambda_n\rightarrow\lambda)$ has been inserted in $H(w_t)$ (lines \ref{tsat:ifbox}-\ref{tsat:endifbox}).  We first show that, for every positive literal $\lambda_j$ in the head of the clause, we have $\lambda_j \in H(w_t) \cup L(w_t)$. If $\lambda_j = p$, then, since $M, w_t \mmodels p$ we have that $p \in L(w_t)$ by the definition of $M$. If $\lambda_j = \boxx^m p$, since we know that $M,w_t\mmodels\lambda_j$, it must be that $M,w_{t+u}\mmodels p$ for each $0 \leq u \leq m$, that is, $p\in L(w_{t+u})$.
We distinguish between two cases: either $w_{t+m} \in W$ or $t+m>|W|-1$.
Suppose, first, that $t+m>|W|-1$. This means that the {\bf for} cycle of the main procedure has been interrupted before reaching the value $md(\varphi)$ (which is the limit also in the case of $\mathbf L=\Tmono$), and that \Call{Shorten}{} has been executed, returning True. Observe that \Call{Shorten}{} may have eliminated several worlds, but immediately before returning True, it has eliminated the world $w_{|W|}$. At precisely this moment, $\boxx p$, $\boxx\boxx p$,\ldots,$\boxx^m p$ have been inserted into $H(w_{|W|-1})$ (line~\ref{sh:addboxes}). Since $t\le |W|-1<t+m$, we have that $(|W|-1)-t\le m$, which, in turn, implies that the last execution of \Call{\Tmono-Saturate}{} must have inserted $\boxx^m p\in H(w_t)$ (lines~\ref{tsat:ifprec}-\ref{tsat:endifprec}). When $w_{t+m} \in W$, a similar argument proves, again, that $\boxx^m p\in H(w_t)$. Either way, $\lambda_j\in L(w_t)$. Since every positive literal in the body of the clause $(\lambda_1\wedge\ldots\wedge\lambda_n\rightarrow\lambda)$ belongs to $H(w_t)$, we have that some execution of \Call{\Tmono-Saturate}{} must have put $\lambda\in H(w_t)$ (lines~\ref{tsat:ifposcl}-\ref{tsat:endifposcl}).
If $\lambda=p$, within the same execution of \Call{\Tmono-Saturate}{} we would have that $p\in L(w_t)$ (lines~\ref{tsat:ifprop}-\ref{tsat:addlprop}), and thus that $M, w_t \mmodels p$, which contradicts the hypothesis that $M$ does not satisfy the formula. Now, if $\lambda=\boxx^m p$, then the executions of \Call{\Tmono-Saturate}{} that take place at the moment of creation of every point $w_{t+u}$, for $0 \leq u \leq m$, put $p\in L(w_{t+u})$. By the reflexivity of $R^\circlearrowright$, this implies that $M, w_t \mmodels \boxx^m p$, which, again, contradicts our hypothesis. In all cases we have a contradiction and thus we have proved that $M, w_0 \mmodels \varphi$.

We show  now that the claim holds for $\Lmono = \Fmono$. Assume that \Call{HornBoxSat}{$\Fmono,\varphi$} returned True, and consider the structure $(W, H, L)$ as it has been built at the end of the execution.
We define a pre-linear \Fmono-model $M$ based on the final structure by distinguishing two cases: either the {\bf for} loop of the main procedure has come to an end, or it has been interrupted and \Call{Shorten}{} has been executed, returning True.
In the former case the final $W$ has $|\varphi| + 1$ worlds, and the structure represents a looping model: $M$ is built on the frame $(W, R)$ where $R = \{(w_{k'}, w_{k''}) \mid k' \leq k''$ or $k' = k'' = |\varphi|\}$ (notice that $R$ self-loops on $w_{|\varphi|}$).
In the latter case (\Call{Shorten}{} is executed) we have that $|W| \leq  |\varphi|$, and the structure represents a finite model with no loops: $M$ is built on the frame $(W, R)$ where $R = \{(w_{i}, w_{j}) \mid i \leq j\}$.
In both cases, for every $w\in W$ we define the valuation as $V(w)=L(w)|_{\mathcal P}$. We want to prove that $M,w_0\mmodels\varphi$, where $w_0$ is the root of the pre-linear \Fmono-model $M$. We prove the correctness of the procedure only for the case of a `looping' pre-linear model.
The case of a non-looping model is simpler and can be proved in the very same way.
Suppose, by way of contradiction, that $M,w_0\not\mmodels\varphi_i$ for some clause $\varphi_i$.
Assume that $\varphi_i=\boxx^s(\lambda_1\wedge\ldots\wedge\lambda_n\rightarrow\lambda)$, and let $w_0,w_1,\ldots$ be the enumeration of the of worlds in $W$. Hence, we have that there exists $w_t \in W$, for some $t \geq s$, such that $M, w_t \mmodels \lambda_1\wedge\ldots\wedge\lambda_n$ but $M, w_t \not\mmodels \lambda$. Since $w_t\in W$, the {\bf for} cycle of the main procedure has run enough times to create $w_t$, and since \Call{\Fmono-Saturate}{} has been executed after that $w_t$ has been created, $(\lambda_1\wedge\ldots\wedge\lambda_n\rightarrow\lambda)$ has been inserted in $H(w_t)$ (lines \ref{fsat:ifbox}-\ref{fsat:endifsuccbox}). We first show that, for every positive literal $\lambda_j$ in the head of the clause, we have $\lambda_j \in H(w_t) \cup L(w_t)$. If $\lambda_j = p$, then, since $M, w_t \mmodels p$ we have that $p \in L(w_t)$ by the definition of $M$. If $\lambda_j = \boxx^m p$, since we know that $M,w_t\mmodels\lambda_j$, it must be the case that either $t + m \leq |\varphi|$ and $M,w_{u}\mmodels p$ for each $u \geq t + m$, or  $t+m > |\varphi|$ and $M,w_{|\varphi|} \mmodels p$ (because of the self-loop in the model). Suppose that we are in the latter case, and let $N = |\varphi| < t + m$: since $M,w_{N} \mmodels p$, by the definition of $M$ we have that this implies that $p \in L(w_{N})$. Hence, the last execution of \Call{\Fmono-saturate}{} must have executed lines \ref{fsat:ifloop}--\ref{fsat:endifloop}, adding $\boxx p$ to $H(w_{N})$.
But, from lines~\ref{fsat:ifloop}--\ref{fsat:endifloop}, we have that also $\boxx\boxx p, \ldots, \boxx^m p$ are in $H(w_N)$, while from lines~\ref{fsat:ifprec}--\ref{fsat:endifprec} we can conclude that $\boxx^m p \in H(w_t)$.
When $t + m \leq |\varphi|$, a similar argument allow us to conclude that also in this case $\lambda_j\in H(w_t)$.
Since every positive literal in the body of the clause $(\lambda_1\wedge\ldots\wedge\lambda_n\rightarrow\lambda)$ belongs to $H(w_t)$, we have that some execution of \Call{\Fmono-Saturate}{} must have put $\lambda\in H(w_t)$ (lines~\ref{fsat:ifposcl}-\ref{fsat:endifposcl}).
If $\lambda=p$, within the same execution of \Call{\Fmono-Saturate}{} we would have that $p\in L(w_t)$ (lines~\ref{fsat:ifprop}-\ref{fsat:addlprop}), and thus that $M, w_t \mmodels p$, which contradicts the hypothesis that $M$ does not satisfy the clause. Conversely, if $\lambda=\boxx^m p$, then, the execution of \Call{\Fmono-Saturate}{} that takes place either at the moment of creation of the point $w_{t+m}$ (when $t+m \leq N$), or at the creation of the point $w_{N}$ (when $t+m > N$) puts $p$ respectively in $L(w_{t+m})$ or $L(w_N)$, which, again, contradicts our hypothesis. In all cases we have a contradiction and thus we have proved that $M, w_0 \mmodels \varphi$.

Finally, the case of $\Lmono = \SFmono$ can be proved by combining the proofs for $\Tmono$ and $\Fmono$.
\end{proof}

The definition of compatibility of a pre-linear non-looping \Kmono-model $M$ with respect to a structure $(W, H, L)$ given above applies also to pre-linear reflexive, transitive, and reflexive and transitive models. The following lemma shows that the saturation procedures preserve compatibility also for \Tmono, \Fmono, and \SFmono.

\begin{lem}\label{lem:l-saturate}
Let $\Lmono \in \{\Tmono, \Fmono, \SFmono\}$, let $M$ be a pre-linear \Lmono-model, and let $(W, H, L)$ be a structure compatible with $M$. Then the structure obtained after the execution of \Call{\Lmono-Saturate}{$W, H, L$} is compatible with $M$.
\end{lem}

\begin{proof}
We first show that the claim holds for $\Lmono = \Tmono$. We proceed as in the proof of Lemma~\ref{lem:k-saturate} to show that after every iteration of the \textbf{while} loop the structure $(W, H, L)$ remains compatible with $M$. From now on, we assume that $M$ has no loops and we identify the worlds of the structure $(W, H, L)$ as $w_0, w_1, \ldots$, and the worlds in $M$ as $w_0', w_1', \ldots$.
Let $w_k \in W$ and $\psi \in H(w_k)$ be respectively the world and the formula selected by the \textbf{let} statement of line~\ref{tsat:let}. If the corresponding $w_i'$ is not a world of $M$, then the property trivially holds. Otherwise, we proceed by following the pseudocode.
Lines~\ref{tsat:ifprec}--\ref{tsat:endifprec} insert $\boxx \psi\in H(w_{k-1})$ if the predecessor $w_{k-1}$ of $w_k$ exists, $\boxx\psi \in Cl(\varphi)$ and $\psi \in H(w_{k-1}) \cup L(w_{k-1})$.
By inductive hypothesis, since $\psi$ was already in $H(w_k)\cup L(w_k)$ and in $H(w_{k-1}) \cup L(w_{k-1})$, it was already true that $M,w_k'\mmodels\psi$ and $M,w_{k-1}'\mmodels\psi$.
Since $M$ is a reflexive pre-linear model, we also have that $M,w_{k-1}' \mmodels \boxx \psi$, as required by the fact that now $\boxx\psi \in H(w_{k-1})\cup L(w_{k-1})$.
The other steps depend on the structure of $\psi$.
If $\psi=p$, then lines~\ref{tsat:ifprop}--\ref{tsat:endifprop}  move $\psi$ from $H(w_k)$ to $L(w_k)$. By inductive hypothesis, since $p$ was already in $H(w_k)\cup L(w_k)$, it was already true that $M,w_k'\mmodels p$. If $\psi=\boxx\xi$, then line~\ref{tsat:boxrefl} inserts $\xi$ in $H(w_k)$. Moreover, if there exists a successor $w_{k+1}$ of $w_k$, then lines \ref{tsat:ifsuccbox}--\ref{tsat:endifsuccbox} move $\boxx\xi$ from $H(w_k)$ to $L(w_k)$, and insert $\xi\in H(w_{k+1})$. By inductive hypothesis, $M,w_k'\mmodels\boxx\xi$: by reflexivity, this means that $M,w_k'\mmodels\xi$ and that, if the successor $w_{k+1}'$ exists in $M$, $M,w_{k+1}'\mmodels\xi$, as required. The proof for the cases when $\psi$ is a clause is similar to Lemma \ref{lem:k-saturate}.

%

Let us prove now that the claim holds for $\Lmono = \Fmono$. We proceed as before to show that after every iteration of the \textbf{while} loop the structure $(W, H, L)$ remains compatible with $M$.
Again, let $w_k \in W$ and $\psi \in H(w_k)$ be respectively the world and the formula selected by the \textbf{let} statement of line~\ref{fsat:let}. If the corresponding $w_k'$ is not a world of $M$, then the property trivially holds.
Otherwise, we proceed by following the pseudocode by considering lines~\ref{fsat:ifloop}--\ref{fsat:endifloop} that inserts $\boxx \psi\in H(w_{N})$ if $\boxx\psi \in Cl(\varphi)$ and $\psi \in H(w_{N}) \cup L(w_{N})$. By inductive hypothesis, since $\psi$ was already in $H(w_N)\cup L(w_N)$, it was already true that $M,w_N'\mmodels\psi$. Since $w_N$ is the last world of a looping model we also have that $M,w_{N}' \mmodels \boxx \psi$, as required by the fact that now $\boxx\psi \in H(w_{N})\cup L(w_{N})$.
Then, lines~\ref{fsat:ifprec}--\ref{fsat:endifprec} insert $\boxx \psi\in H(w_{k-1})$ if the predecessor $w_{k-1}$ of $w_k$ exists, $\boxx\psi \in Cl(\varphi)$ and $\boxx\psi \in H(w_{k}) \cup L(w_{k})$.
By inductive hypothesis, since $\psi$ and $\boxx\xi$ were already in $H(w_k)\cup L(w_k)$, it was already true that $M,w_k'\mmodels\psi$ and $M,w_{k}'\mmodels\boxx\psi$.
Since $M$ is a transitive pre-linear model, we also have that $M,w_{k-1}' \mmodels \boxx \psi$, as required by the fact that now $\boxx\psi \in H(w_{k-1})\cup L(w_{k-1})$.
The other steps depend on the structure of $\psi$.
If $\psi=p$, then lines~\ref{fsat:ifprop}--\ref{fsat:endifprop} move $\psi$ from $H(w_k)$ to $L(w_k)$.  By inductive hypothesis, since $p$ was  in $H(w_k)\cup L(w_k)$, it was already true that $M,w_k'\mmodels p$.
In the case $\psi=\boxx\xi$, if there exists a successor $w_{k+1}$ of $w_k$, then lines \ref{fsat:ifsuccbox}--\ref{fsat:endifsuccbox} move $\boxx\xi$ from $H(w_k)$ to $L(w_k)$, and insert $\xi, \boxx\xi\in H(w_{k+1})$. By inductive hypothesis, $M,w_{k}'\mmodels\boxx\xi$, which imply  that $M,w_{k+1}'\mmodels\xi$ and $M_{k+1}'\mmodels\boxx\xi$, as we wanted. The rest of the cases are treated as in Lemma \ref{lem:k-saturate}.
\end{proof}

\noindent Once we have proved that \Call{\Lmono-Saturate}{} preserves compatibility, proving that \Call{HornBoxSat}{} is complete becomes straightforward.

\begin{lem}\label{lem:l-completeness}
Let $\Lmono \in \{\Tmono, \Fmono, \SFmono\}$ and let $\varphi$ be a \LHornBox-formula. If $\varphi$ is satisfiable then \Call{HornBoxSat}{$\Lmono,\varphi$} returns True.
\end{lem}

\begin{proof}
We can proceed as in the proof of Lemma~\ref{lem:k-completeness}. We assume that $\varphi$ is satisfiable, and
thanks to Theorem~\ref{th:prelinear}, we know that there exist some pre-linear models of size bounded either by $md(\varphi)+1$ (for \Tmono) or by $|\varphi|$ (for \Fmono and \SFmono). Since our procedure starts iterating from a structure $(W, H, L)$ compatible with $M$, and since \Call{\Lmono-Saturate}{} preserves compatibility (Lemma~\ref{lem:l-saturate}), we can prove, as in Lemma~\ref{lem:k-completeness}, that \Call{HornBoxSat}{$\Lmono,\varphi$} builds the longest model $M$ of $\varphi$, and therefore returns True.
\end{proof}

The following theorem summarises the above results and shows that the procedure is correct, complete and polynomial for all the considered languages.

\begin{thm}\label{th:complexity}
Let $\Lmono \in \{\Kmono, \Tmono, \Fmono, \SFmono\}$ and let $\varphi$ be a \LHornBox-formula. Then, the procedure \Call{HornBoxSat}{$\mathbf L,\varphi$} is correct, complete and terminates within $O(poly(|\varphi|))$ time.
\end{thm}

\begin{proof}
Correctness and completeness of \Call{HornBoxSat}{} follows from Lemma~\ref{lem:k-correctness} and \ref{lem:k-completeness} when $\Lmono = \Kmono$ and from Lemma~\ref{lem:l-correctness} and \ref{lem:l-completeness} when $\Lmono \in \{\Tmono, \Fmono, \SFmono\}$.

As for the time complexity of the procedure, we observe that the most external cycle runs $O(|\varphi|)$ times in the worst case. The dimension of the candidate models grows from 1 to $O(|\varphi|)$ points, and, then, down to 1 again; in the $H$ component of each point there are, at most, $O(|\varphi|)$ formulas on which \Call{Saturate}{} has effect. Therefore, the total time spent is $O(|\varphi|)\cdot 2\cdot\Sigma_{i=1}^{O(|\varphi|)} i=O(|\varphi|^3)$, which is polynomial in $|\varphi|$.
\end{proof}

We conclude this section by studying the particular case of the fragment $\KmonoCoreBox$. The results obtained so far allow us only to fix the complexity of the satisfiability problem in these cases between \NLOG\ and \Pt, the lower bound being a consequence of the \NLOG-completeness of the satisfiability problem for binary propositional clauses (i.e., the 2SAT problem~\cite
{papa2003}). We want to prove that the satisfiability problem for $\KmonoCoreBox$, the satisfiability problem is, in fact, \NLOG-complete; to this end, we extend the algorithm for 2SAT. The standard algorithm to solve a 2SAT instance exploits the properties of the \emph{implication graph} of the formula. An implication graph is a directed graph in which there is one vertex per propositional variable or negated variable, and an edge connecting one vertex to another whenever the corresponding variables are related by an implication in the formula. A propositional binary formula is unsatisfiable if and only if there exists a path in the implication graph that starts from a variable $p$, leads to its negation $\neg p$ and then goes back to $p$. The \NLOG\ algorithm simply guesses the variable $p$ and search nondeterministically for the chain of implications that corresponds to such a contradictory path. An alternative, indirect way to obtain this result would be to explore the characteristics of SLD-resolution~\cite{DBLP:journals/fuin/Nguyen03} under the Horn box restriction of \Kmono.

\medskip

In the implication graph of a \KmonoCoreBox-formula the vertexes are pairs where the first component is a positive literal $\lambda$ or its negation $\neg\lambda$, and the second component is a natural number representing the depth of the world where $\lambda$ (resp.,$\neg\lambda$) holds, measured as the distance from the root of the model; it is worth to recall that \KmonoCoreBox, as \KmonoHornBox, enjoys the pre-linear model property, so that the notion of distance from the root makes perfect sense. The edges in the graph correspond either to implications in the formula (as in the propositional case), or to {\em jumps} from a world at some depth $d$ to the world at depth $d+1$ or $d-1$. Thus, given $\varphi\in\KmonoHornBox$, where $\varphi_1,\ldots,\varphi_l$, and given a natural number $D\le md(\varphi)$,  we define the \emph{$D$-implication graph of $\varphi$} as a directed graph $G_\varphi^D = (V, E)$ where $V = \{(\lambda, d) \mid \lambda \in Cl(\varphi)$ or $\neg\lambda \in Cl(\varphi)$ and $0 \leq d < D\}$, and:
\begin{inparaenum}
	\item $((\lambda, d),(\lambda', d)),((\neg\lambda', d),(\neg\lambda, d))\in E$ if and only if  $\varphi_i=\boxx^d(\neg\lambda \lor \lambda')$ for some $i$;
	\item $((\boxx\lambda, d),(\lambda, d+1))\in E$ if and only if $d<D-1$;
	\item $((\lambda, d),(\boxx\lambda, d-1))\in E$ if and only if $d>0$;
	\item $((\neg\boxx\lambda, d),(\neg\lambda, d+1))\in E$ if and only if $d<D-1$; and
	\item $((\neg\lambda, d),(\neg\boxx\lambda, d-1))\in E$ if and only if $d>0$.
\end{inparaenum}

Notice that, in this case, we treat \KmonoCoreBox-clauses as disjunctions, instead of implications; in this way, clauses that contain a single positive (resp., negative) literal $\lambda$ (resp., $\neg \lambda$) are dealt with as disjunctions of the type $(\lambda\vee\bot)$ (resp., $(\neg\lambda\vee\bot)$). Moreover, for technical reasons, we add to each formula the clauses $\boxx^s(\bot\rightarrow\top)$ for each $s\le md(\varphi)$. The main property of the $D$-implication graph is that the formula on which built is unsatisfiable on (pre-linear) models of length $D$ if and only if there exists a cycle $\pi$ that starts and ends in some vertex $(\lambda, d)$ and visits the vertex $(\neg\lambda, d)$; in this case, $\pi$ is called {\em contradictory cycle}.

\begin{lem}\label{lem:k_horn_circuit1}
Let $\varphi$ be a formula of \KmonoHornBox. Then, if there exists a contradictory cycle $\pi$ in $G_\varphi^D$, then $\varphi$ is unsatisfiable on a linear model with precisely $D$ distinct worlds.
%
%
\end{lem}

\begin{proof}
Let $\varphi$ be a formula of \KmonoHornBox, and suppose that there exists a vertex $(\lambda,d)$ and a contradictory cycle $\pi=(\lambda_1,d_1),(\lambda_2,d_2),\ldots,(\lambda_{|\pi|},d_{|\pi|})$ in the implication graph $G_\varphi^D$ such that $(\lambda, d) = (\lambda_1,d_1)$. Suppose by contradiction that $\varphi$ is satisfiable on a pre-linear model $M$ with precisely $D$ distinct worlds such that $M, w_0 \mmodels \varphi$. Two cases may arise: either $M, w_d \mmodels \lambda$ or $M, w_d \mmodels \neg\lambda$. We consider here only the former case (the latter one can be proved in the very same way), and we show that for every vertex $(\lambda_i,d_i)\in\pi$ it is the case that $M,w_{d_i}\mmodels\lambda_i$. We proceed by induction on $i$. When $i = 1$ then $(\lambda_1, d_1) = (\lambda, d)$, and since $M, w_d \mmodels \lambda$ the property trivially holds. Now, suppose that the property is true for the vertex $(\lambda_i,d_i)$, and consider the $i+1$-th vertex $(\lambda_{i+1}, d_{i+1})$; three cases may arise:


\medskip

\begin{compactitem}
\item $d_{i+1} = d_i$: in this case we have that the clause $\boxx^{d_i}(\neg\lambda_i \lor \lambda_{i+1})$ is a conjunct of $\varphi$. Since $M, w_0 \mmodels \boxx^{d_i}(\neg\lambda_i \lor \lambda_{i+1})$ and $w_i$ is at distance $d_i$ form $w_0$, we have that $M, w_{d_i} \mmodels \neg\lambda_i \lor \lambda_{i+1}$. Since by  inductive hypothesis we have that $M, w_{d_i} \mmodels \lambda_i$, then, as we wanted, $M, w_{d_{i+1}} \mmodels \lambda_{i+1}$.

	\item $d_{i+1} = d_i + 1$: in this case we have that either $\lambda_{i} = \boxx\lambda$ or $\lambda_{i} = \neg\boxx\lambda$. In the first case, the inductive hypothesis is that $M, w_{d_i} \mmodels \boxx\lambda$, which is to say that $M,w_{d_i+1}\mmodels\lambda$, that is, $M,w_{d_i+1}\mmodels\lambda_{i+1}$, as we wanted. In the second case the inductive hypothesis is that $M, w_{d_i} \mmodels \neg\boxx\lambda$, which is to say that $M,w_{d_i+1}\mmodels\neg\lambda$, that is, $M,w_{d_i+1}\mmodels\lambda_{i+1}$, as we wanted.

		\item $d_{i+1} = d_i - 1$: in this case we have that either $\lambda_{i+1} = \boxx\lambda$ or $\lambda_{i+1} = \neg\boxx\lambda$. In the first case,
it must be that $\lambda_i=\lambda$; therefore, by inductive hypothesis, $M, w_{d_i} \mmodels \lambda$, which is to say that $M,w_{d_i-1}\mmodels\boxx\lambda$ (since $M$ is a pre-linear model), that is, $M,w_{d_i-1}\mmodels\lambda_{i+1}$, as we wanted. In the second case it must be that $\lambda_i=\neg\lambda$; therefore, by inductive hypothesis, $M, w_{d_i} \mmodels \neg\lambda$, which is to say that $M,w_{d_i-1}\mmodels\neg\boxx\lambda$, that is, $M,w_{d_i-1}\mmodels\lambda_{i+1}$, as we wanted.
\end{compactitem}

\medskip

\noindent Since both $(\lambda, d)$ and $(\neg\lambda, d)$ belong to $\pi$, by the above property we have that $M, w_d \mmodels \lambda$ and $M, w_d \mmodels\neg\lambda$, and a contradiction is found. Hence, $\varphi$ cannot be satisfied in a pre-linear model with precisely $D$ distinct worlds.
\end{proof}

\begin{lem}\label{lem:k_horn_circuit2}
Let $\varphi$ be a formula of \KmonoHornBox. Then, if there is no contradictory cycle $\pi$ in $G_\varphi^D$, then $\varphi$ is satisfiable on a linear model with precisely $D$ distinct worlds.
\end{lem}

\begin{proof}
Let $\varphi$ be a formula of \KmonoHornBox, and suppose that no contradictory cycle is in $G_\varphi^D$. We show how to build a pre-linear model for $\varphi$ with exactly $D$ distinct worlds. Let $W=\{w_0,\ldots,w_{D-1}\}$, and let $(w_i,w_{i+1})\in R$ for each $i\le D-1$. Now, we have to show how to build a consistent evaluation for each world. As a preliminary step, we show the following property on the $D$-implication graph for $\varphi$: if $(\lambda,d)\rightsquigarrow (\lambda',d')$ (where $\rightsquigarrow$ means that there is a path from $(\lambda,d)$ to $(\lambda',d')$), then $(\neg\lambda',d')\rightsquigarrow (\neg\lambda,d)$. We can prove this by induction on the length of the path. The base case is $|\pi|=2$ (the case in which $|\pi|=1$ is trivial), where three cases may arise:

\medskip

\begin{compactitem}
\item $d = d'$. In this case, we have that the clause $\boxx^{d}(\neg\lambda \lor \lambda')$ is a conjunct of $\varphi$; but the rules for the construction of the implication graph include the edge $(\neg\lambda',d),(\neg\lambda,d)$ as well, and, thus, we have the claim.

\item $d' = d + 1$. In this case we have that either $\lambda = \boxx\lambda'$, or $\lambda=\neg\boxx\overline\lambda$ and $\lambda'=\neg\overline \lambda$. In the first case, the rules for the construction of the implication graph include the edge $(\neg\lambda',d+1),(\neg\boxx\lambda',d)$, implying the claim; in the second case the rules for the construction of the implication graph include the edge $(\overline\lambda,d+1),(\boxx\overline\lambda,d)$, implying, again, the claim.

		\item $d'=d-1$. This case is completely symmetric to the above one.
\end{compactitem}

\medskip

\noindent The inductive case can be proved by using precisely the same rules; the property holds up to a certain length $i$, and, then, a single application of the above cases leads to prove that it holds for the length $i+1$. Now, we can build a function $V^*:W\rightarrow 2^{Cl(\varphi)}$, as follows. Starting from the world $w_0$, we choose $\lambda\in Cl(\varphi)$, and we set $\lambda\in V(w_0)$ if and only if there is no path from $(\lambda,0)$ to $(\neg\lambda,0)$ in $G_\varphi^D$. Then, for each $(\lambda',d)$ reachable from $(\lambda,0)$ in $G_\varphi^D$, we set $\lambda'\in V(w_d)$. We now repeat the same process for each $\lambda'\in Cl(\varphi)$ for which a decision for $w_0$ has not been taken yet. Finally, we repeat all three steps for each $w_d\in W$. At this point, we can synthesize an evaluation $V:W\rightarrow 2^{\mathcal P}$ from $V^*$ by simply projecting $V^*$ over the propositional letters. Since the above property implies that such a labeling is contradiction-free, it is easy to see that $M,w_0\mmodels\varphi$.
\end{proof}

\medskip

The above lemmas can be then immediately used to devise a non-deterministic algorithm for $\KmonoCoreBox$: we simply guess the length of the candidate model $D\le md(\varphi)$, and we check that no contradictory cycle exists, which proves that $\varphi$ has a pre-linear model of dimension $D$. Then, the following result holds.

\begin{thm}\label{th:khornbox_nlog}
The satisfiability problem for $\KmonoCoreBox$ is \NLOG-complete.
\end{thm}

\begin{proof}
The fact that the satisfiability problem for \LCoreBox\ is in \NLOG\ is a consequence of the above argument and the fact that checking the non-existence of a cycle in a direct graph is \NLOG. To avoid the explicit construction of the entire graph, that would require a polynomial amount of space, the \NLOG\ algorithm simulates the construction and checks for the non-existence of contradictory cycles on-the-fly. In~\cite{papa2003} \NLOG-hardness of 2SAT is proved by reducing the unreachability problem on acyclic graphs to 2SAT. Since the formula used for the reduction includes only clauses of the form $(\neg p \lor q)$, $p$, and $\neg p$, it is indeed a formula of the core fragment of propositional logic. Hence, \NLOG-hardness of \LCoreBox\ immediately follows.
\end{proof}

\medskip

The complexity results of this of this section can be seen in Fig.~\ref{fig:complexity}.

\begin{figure}[t!]

\begin{tikzpicture}[>=latex,line join=bevel,scale=0.62,thick]
  \node (Core) at (336.0bp,106.0bp) [draw,draw=none] {$\LCore$};
  \node (Bool) at (335.0bp,250.0bp) [draw,draw=none] {$\LBool$};
  \node (CoreBox) at (285.0bp,34.0bp) [draw,draw=none] {$\LCoreBox$};
  \node (Krom) at (177.0bp,178.0bp) [draw,draw=none] {$\LKrom$};
  \node (HornBox) at (500.0bp,106.0bp) [draw,draw=none] {$\LHornBox$};
  \node (Horn) at (494.0bp,178.0bp) [draw,draw=none] {$\LHorn$};
  \node (KromDia) at (172.0bp,106.0bp) [draw,draw=none] {$\LKromDia$};
  \node (HornDia) at (606.0bp,106.0bp) [draw,draw=none] {$\LHornDia$};
  \node (KromBox) at (63.0bp,106.0bp) [draw,draw=none] {$\LKromBox$};
  \node (CoreDia) at (389.0bp,34.0bp) [draw,draw=none] {$\LCoreDia$};
  \draw [->,solid] (Krom) ..controls (175.21bp,151.98bp) and (174.55bp,142.71bp)  .. (KromDia);
  \draw [->,solid] (Horn) ..controls (496.14bp,151.98bp) and (496.94bp,142.71bp)  .. (HornBox);
  \draw [->,solid] (Horn) ..controls (535.91bp,150.81bp) and (553.92bp,139.55bp)  .. (HornDia);
  \draw [->,solid] (KromDia) ..controls (215.2bp,85.116bp) and (222.26bp,82.327bp)  .. (229.0bp,80.0bp) .. controls (273.42bp,64.667bp) and (288.24bp,65.91bp)  .. (CoreDia);
  \draw [->,solid] (HornDia) ..controls (566.65bp,85.093bp) and (560.19bp,82.312bp)  .. (554.0bp,80.0bp) .. controls (517.07bp,66.216bp) and (474.06bp,54.626bp)  .. (CoreDia);
  \draw [->,solid] (Bool) ..controls (279.09bp,224.23bp) and (246.9bp,209.97bp)  .. (Krom);
  \draw [->,solid] (KromBox) ..controls (104.67bp,85.136bp) and (111.49bp,82.341bp)  .. (118.0bp,80.0bp) .. controls (155.13bp,66.659bp) and (198.15bp,55.157bp)  .. (CoreBox);
\draw [->,solid] (Krom) ..controls (234.17bp,151.83bp) and (263.35bp,138.99bp)  .. (Core);
  \draw [->,solid] (Horn) ..controls (438.57bp,152.44bp) and (408.8bp,139.25bp)  .. (Core);
  \draw [->,solid] (Core) ..controls (317.4bp,79.474bp) and (310.12bp,69.483bp)  .. (CoreBox);
 \draw [->,solid] (Bool) ..controls (391.5bp,224.12bp) and (425.12bp,209.32bp)  .. (Horn);
  \draw [->,solid] (Krom) ..controls (134.21bp,150.72bp) and (115.67bp,139.34bp)  .. (KromBox);
  \draw [->,solid] (Core) ..controls (355.39bp,79.389bp) and (363.05bp,69.277bp)  .. (CoreDia);
  \draw [->,solid] (HornBox) ..controls (456.57bp,85.077bp) and (449.63bp,82.322bp)  .. (443.0bp,80.0bp) .. controls (400.18bp,64.998bp) and (386.05bp,65.476bp)  .. (CoreBox);
%
%
\node[rotate=2.5] at (600bp,207bp) {\footnotesize{$\in\Pt$}~\cite{nguyen2004complexity}, $\mathbf L=\SFive$};
\node[rotate=2.5] at (600bp,155bp)  {\begin{tabular}{l}\footnotesize{$\PSpace$-complete}~\cite{ChenLin94},\\ \quad\quad\footnotesize{$\mathbf L=\Kmono,\Tmono,\Fmono,\SFmono$}\end{tabular}};
\node[rotate=2.5] at (570bp,245bp) {\footnotesize{$\NP$-complete}~\cite{nguyen2004complexity}, $\mathbf L=\SFive$};
\node[rotate=6] at (120bp,20bp) {\begin{tabular}{l}\footnotesize{$\NLOG$-hard, $\in \Pt$},\\
                                                   \footnotesize{$\mathbf L=\Tmono,\Fmono,\SFmono$}, \\
                                                   \footnotesize{$\NLOG$-complete, $\mathbf L=\Kmono$}
                                  \end{tabular}};
\node at (535bp,30bp) {\begin{tabular}{l}\footnotesize{$\Pt$-complete},\\ \footnotesize{$\mathbf L=\Kmono,\Tmono,\Fmono,\SFmono$}\end{tabular}};
\draw [dashed] plot [smooth] coordinates {(280bp,240bp) (380bp,140bp) (675bp,140bp) };
\draw [dashed] plot [smooth] coordinates {(415bp,20bp) (478bp,115bp) (520bp,115bp) (650bp,20bp) };
\draw [dashed] plot [smooth] coordinates {(250bp,260bp) (320bp,220bp) (655bp,235bp) };
\draw [dashed] plot [smooth] coordinates {(230bp,20bp) (370bp,190bp) (655bp,220bp) };
\draw [dashed] plot [smooth] coordinates {(50bp,40bp) (270bp,60bp) (340bp,20bp) };
\end{tikzpicture}
	\caption{Complexity results}
   \label{fig:complexity}
\end{figure}

\section{Conclusions}

In this paper, we considered the expressive power and complexity of sub-propositional fragments of  the basic modal logics of all directed graphs (\Kmono) and its axiomatic reflexive (\Tmono), transitive (\Fmono), and reflexive and transitive (\SFmono) extensions. The relative expressive power and the complexity results of the sub-propositional fragments of modal logic studied in this paper is depicted in Figures~\ref{fig:exprpowerComplete} and \ref{fig:complexity}, respectively. In most cases relative expressivity coincides with syntactical containment, with the notable exception of the Krom fragments, that are expressively equivalent, but not weakly expressively equivalent. Because of our very general approach for comparing the expressive power of languages, most of our result can be transferred to other sub-propositional modal logics such as the fragments of \LTL without Since and Until studied in~\cite{DBLP:conf/lpar/ArtaleKRZ13} and the sub-propositional fragments of \HS~\cite{artale2015,bresolin2017horn,DBLP:conf/jelia/BresolinMS14}. To the best of our knowledge, this is the first work where sub-Krom and sub-Horn fragments of these modal logics have been considered. Concerning complexity, starting from the known result that every modal logic between \Kmono and \SFmono is still \PSpace-complete even under the Horn restriction~\cite{nguyen2004complexity} (in opposition to the case of \SFive, which goes from \NP-complete to \Pt-complete), we proved that eliminating the use of diamonds in Horn formulas reduces the complexity of the satisfiability problem in all considered cases; some of our results can be also proved by using the machinery in~\cite{Nguyen2000,Nguyen2016}, which, on the other hand, does not immediately provide an implementable satisfiability-checking procedure. We proved that the satisfiability problem becomes \Pt-complete in the Horn fragment without diamonds, between \NLOG\ and \Pt\ in the core fragment without diamonds (and \NLOG-complete in the special case of \KmonoCoreBox), and we devised a fast satisfiability-checking algorithms for all considered cases. 

\medskip

While this study has been inspired from similar results in the temporal case, such as for fragments of \LTL\ without Since and Until and fragment of \HS, several natural questions concerning sub-Horn temporal logics remain open, and our results can certainly offer a solid starting point to their tackling. To mention one interesting case, consider \HS under the Horn restriction without diamonds: it is known that it is still undecidable in the discrete case~\cite{bresolin2017horn}, but such a proof makes extensive use of many different temporal operators; what would happen by considering simple fragments such as $\mathrm A$ or $\mathrm D$ (i.e., the fragments with a single modal operator corresponding to, respectively, Allen's relation {\em meets} and Allen's relation {\em during}) is unclear. Our technique, based on the intrinsic inability of \SFmono to force the non-linearity of a model under the considered restriction, seems, at least in principle, applicable to such a case. Other interesting cases include the behaviour of \LTL with Since/Until under the restrictions considered in this paper.

\section*{Acknowledgements}
The authors acknowledge the support from the Italian INdAM -- GNCS Project 2018 ``\emph{Formal methods for verification and  synthesis of discrete and hybrid systems}'' (D. Bresolin and G. Sciavicco), and the Spanish Project TIN15-70266-C2-P-1 (E. Mu\~noz-Velasco).


%
%

\bibliographystyle{alpha}
\bibliography{biblio}




\end{document}

%% file: macro.tex
\newcommand{\LmonoBool}
{\ensuremath{\operatorname{\mathbf{L^{Bool}}}}
\xspace}
\newcommand{\Kmono}{\ensuremath{\operatorname{\mathbf{K}}}\xspace}
\newcommand{\Tmono}{\ensuremath{\operatorname{\mathbf{T}}}\xspace}
\newcommand{\Fmono}{\ensuremath{\operatorname{\mathbf{K4}}}\xspace}
\newcommand{\SFmono}{\ensuremath{\operatorname{\mathbf{S4}}}\xspace}
\newcommand{\SFive}{\ensuremath{\operatorname{\mathbf{S5}}}\xspace}
\newcommand{\KmonoHorn}{\ensuremath{\operatorname{\mathbf{K^{Horn}}}}\xspace}
\newcommand{\LmonoHorn}{\ensuremath{\operatorname{\mathbf{L^{Horn}}}}\xspace}
\newcommand{\TmonoHorn}{\ensuremath{\operatorname{\mathbf{T^{Horn}}}}\xspace}
\newcommand{\FmonoHorn}{\ensuremath{\operatorname{\mathbf{K4^{Horn}}}}\xspace}
\newcommand{\SFmonoHorn}{\ensuremath{\operatorname{\mathbf{S4^{Horn}}}}\xspace}
\newcommand{\KmonoHornBox}{\ensuremath{\operatorname{\mathbf{K^{Horn,\Box}}}}\xspace}
\newcommand{\TmonoHornBox}{\ensuremath{\operatorname{\mathbf{T^{Horn,\Box}}}}\xspace}
\newcommand{\SFmonoHornBox}{\ensuremath{\operatorname{\mathbf{S4^{Horn,\Box}}}}\xspace}
\newcommand{\FmonoHornBox}{\ensuremath{\operatorname{\mathbf{K4^{Horn,\Box}}}}\xspace}
\newcommand{\LmonoHornBox}{\ensuremath{\operatorname{\mathbf{L^{Horn,\Box}}}}\xspace}

\newcommand{\LmonoHornDia}{\ensuremath{\operatorname{\mathbf{L^{Horn,\Diamond}}}}\xspace}

\newcommand{\LmonoKrom}{\ensuremath{\operatorname{\mathbf{L^{Krom}}}}\xspace}

\newcommand{\LmonoKromBox}{\ensuremath{\operatorname{\mathbf{L^{Krom,\Box}}}}\xspace}

\newcommand{\LmonoKromDia}{\ensuremath{\operatorname{\mathbf{L^{Krom,\Diamond}}}}\xspace}

\newcommand{\LmonoCore}{\ensuremath{\operatorname{\mathbf{L^{core}}}}\xspace}
\newcommand{\KmonoCoreBox}{\ensuremath{\operatorname{\mathbf{K^{core,\Box}}}}\xspace}
\newcommand{\LmonoCoreBox}{\ensuremath{\operatorname{\mathbf{L^{core,\Box}}}}\xspace}

\newcommand{\LmonoCoreDia}{\ensuremath{\operatorname{\mathbf{L^{core,\Diamond}}}}\xspace}

\newcommand{\Lmono}{\ensuremath{\operatorname{\mathbf{L}}}\xspace}
\newcommand{\LBool}{\ensuremath{\operatorname{\mathbf{L^{Bool}}}}\xspace}
\newcommand{\LHorn}{\ensuremath{\operatorname{\mathbf{L^{Horn}}}}\xspace}
\newcommand{\LHornBox}{\ensuremath{\operatorname{\mathbf{L^{Horn,\Box}}}}\xspace}
\newcommand{\LHornDia}{\ensuremath{\operatorname{\mathbf{L^{Horn,\Diamond}}}}\xspace}
\newcommand{\LKrom}{\ensuremath{\operatorname{\mathbf{L^{Krom}}}}\xspace}
\newcommand{\LKromBox}{\ensuremath{\operatorname{\mathbf{L^{Krom,\Box}}}}\xspace}
\newcommand{\LKromDia}{\ensuremath{\operatorname{\mathbf{L^{Krom,\Diamond}}}}\xspace}
\newcommand{\LCore}{\ensuremath{\operatorname{\mathbf{L^{core}}}}\xspace}
\newcommand{\LCoreBox}{\ensuremath{\operatorname{\mathbf{L^{core,\Box}}}}\xspace}
\newcommand{\LCoreDia}{\ensuremath{\operatorname{\mathbf{L^{core,\Diamond}}}}\xspace}

\newcommand{\THornBox}{\ensuremath{\operatorname{\mathbf{T^{Horn,\Box}}}}\xspace}
\newcommand{\Gr}{\ensuremath{\mathbf{GR(1)}}\xspace}

\newcommand{\FHornBox}{\ensuremath{\operatorname{\mathbf{K4^{Horn,\Box}}}}\xspace}
\newcommand{\SFHornBox}{\ensuremath{\operatorname{\mathbf{S4^{Horn,\Box}}}}\xspace}

\newcommand{\diax}{{\ensuremath{\Diamond}\xspace}}
\newcommand{\boxx}{{\ensuremath{\Box}\xspace}}
\newcommand{\mmodels}{\Vdash}
\newcommand{\eqclass}[1]{\llbracket{#1}\rrbracket}
\newcommand{\Rx}{\mathrel{R}}
\newcommand{\Rprod}{\mathrel{R_{R_1\times R_2}}}

\newcommand{\NP}{\textsc{NP}}
\newcommand{\Pt}{\textsc{P}}

\newcommand{\NLOG}{\textsc{NLogSpace}}
\newcommand{\PSpace}{\textsc{PSpace}}

\newcommand{\HS}{\ensuremath{\mathbf{HS}}\xspace}

\newcommand{\LTL}{\ensuremath{\mathbf{LTL}}\xspace}
\newcommand{\QBF}{\ensuremath{\mathbf{QBF}}\xspace}

\newcommand{\mpequiv}{\equiv^p}
\newcommand{\meequiv}{\equiv^e}
\newcommand{\wlexp}{\prec^w}
\newcommand{\slexp}{\prec^s}
